\documentclass[10pt,british,english,leqno]{siamltex1213}
\usepackage[T1]{fontenc}
\usepackage[latin9]{inputenc}
\usepackage{amsmath}
\usepackage{amssymb}
\usepackage{graphicx}

\makeatletter

\usepackage[left=3.0cm,right=3.0cm,top=2.2cm,bottom=2.2cm]{geometry}

\title{The impact of startup costs and the grid operator on the power price equilibrium\thanks{We thank the Oxford-Man Institute for providing historical prices used to calibrate our model, and ELEXON for providing historical data about the Balancing Mechanism used to determine physical characteristics of the power plants connected to the UK power grid. }} 

\author{Miha Troha\thanks{Mathematical Institute, Oxford University, Andrew Wiles Building, Radcliffe Observatory Quarter, Woodstock Road, Oxford OX2 6GG, United Kingdom, \email{troha@maths.ox.ac.uk}. This author was supported through grants from the Slovene human resources development and scholarship fund, and the Oxford-Man Institute.} \and Raphael Hauser\thanks{Mathematical Institute, Oxford University, Andrew Wiles Building, Radcliffe Observatory Quarter, Woodstock Road, Oxford OX2 6GG, United Kingdom, \email{hauser@maths.ox.ac.uk}. Associate Professor in Numerical Mathematics, and Tanaka Fellow in Applied Mathematics at Pembroke College, Oxford. This author was supported through grant EP/H02686X/1 from the Engineering and Physical Sciences Research Council of the UK.}}

\usepackage{tikz}
\usepackage{verbatim}
\usepackage{varwidth}
\usepackage{ragged2e}
\usetikzlibrary{shapes,arrows,positioning,fit,calc,backgrounds,decorations.shapes}

\tikzset{decorate sep/.style 2 args=
{decorate,decoration={shape backgrounds,shape=circle,shape size=#1,shape sep=#2}}}

\makeatother

\usepackage{babel}
\begin{document}
\maketitle
\begin{abstract}
In this paper we propose a quadratic programming model that can be
used for calculating the term structure of electricity prices while
explicitly modeling startup costs of power plants. In contrast to
other approaches presented in the literature, we incorporate the startup
costs in a mathematically rigorous manner without relying on ad hoc
heuristics. Moreover, we propose a tractable approach for estimating
the startup costs of power plants based on their historical production.
Through numerical simulations applied to the entire UK power grid,
we demonstrate that the inclusion of startup costs is necessary for
the modeling of electricity prices in realistic power systems. Numerical
results show that startup costs make electricity prices very spiky.
In the second part of the paper, we extend the initial model by including
the grid operator who is responsible for managing the grid. Numerical
simulations demonstrate that robust decision making of the grid operator
can significantly decrease the number and severity of spikes in the
electricity price and improve the reliability of the power grid.
\end{abstract}

\section{Introduction}

More than two decades ago, electricity markets started the transition
from a regulated market with a single utility company to a fully competitive
market. This introduced a need for a development of financial models
that would help us to understand the behavior of electricity prices
and manage the risk. High uncertainty in the electricity demand and
fuel prices requires robust models, so that low electricity prices
and a reliable delivery of electricity can be achieved.

Electricity markets are changing extremely quickly, often faster than
any other financial markets. High pressure on decarbonization has
led to new market design and policies. New, intermittent, renewable
sources are connected to the electricity grid almost on a daily basis.
Smart grids, together with the battery storage and demand response,
are making their way into market. All these inventions have an impact
on the electricity price and its behavior. The pace of new inventions
makes the risk management and decision making in electricity markets
very challenging.

In the literature, there exist three approaches to the modeling of
electricity prices. Traditionally, electricity prices have been modeled
by so called non-structural approaches. These approaches attempt to
model electricity prices directly without explicitly considering the
fundamental factors that drive such behavior. \cite{lucia2000electricity}
investigated the statistical properties of the electricity prices
at the Nordic Power Exchange. The suitability of one and multi-factor
Ornstein-Uhlenbeck processes for modeling the spot as well as the
log spot price was examined. As pointed out in this work, none of
these models are able to capture the spikes in the electricity price.
Thus, various other models that combine the Ornstein-Uhlenbeck process
with a pure jump-process (see \cite{hambly2009modelling} for example)
or more general Levy process (see \cite{meyer-brandis2008multifactor}
and \cite{garcia2011estimation} for example) were proposed. A one-factor
model in \cite{clewlow1999valuing} and multi-factor term-structure
model in \cite{clewlow1999amultifactor} are the first that produce
prices that are consistent with observable forward prices. While non-structural
models are widely used for the short-term risk management as well
as electricity derivatives pricing purposes in practice, they do not
cater well for longer-term modeling purposes, where the impact of
new inventions must be included. They must be frequently recalibrated
to reflect the changes in the markets. 

Structural approaches for modeling electricity prices capture some
of the fundamental factors of the electricity market. The supply and
demand stack was first used to model electricity prices in \cite{barlow2002adiffusion}.
This idea was extended by \cite{howison2009stochastic} and by \cite{carmona2013electricity},
where an exponential supply and demand stack was modeled as a function
of the underlying fuels such as natural gas and coal.

The third, game theoretic, approach models the electricity market
even more closely. The disastrous events that happened in California
in 2001 confirmed that some physical properties of power plants such
as ramp-up and ramp-down constants, and market design together with
the transmission lines play a vital role in the behavior of electricity
prices. The first game theoretic model for modeling the electricity
prices was proposed in \cite{bessembinder2002equilibrium}, where
a unique relation between a forward and a spot price is given in a
two-stage market with one producer and one consumer, who each want
to maximize their mean-variance objective function. This model was
extended to a multistage setting in \cite{buhler2009valuation} and
\cite{buhler2009riskpremia}, and to any convex risk measure in \cite{demaeredaertrycke2012liquidity}.
\cite{troha2014theexistence} further extended the work of \cite{buhler2009riskpremia}
to a setting with more than one producer and consumer, who optimize
their mean-variance objective functions. In contrast to other game
theoretic models, capacity and ramp-up and ramp-down constraints of
power plants are included. By modeling the profit of power plants
as a difference between the power price and fuel costs together with
emissions obligations, this work also incorporates ideas from the
structural approach. As in \cite{clewlow1999amultifactor} and \cite{clewlow1999valuing},
the model is consistent with observable fuel and emission prices.
\cite{troha2014calculation} applied this model to calculate the electricity
prices in the UK by taking into account the entire power grid consisting
of a few hundred power plants. Numerical simulations show that this
model has a tendency to underestimate spot prices during the peak
hours and to overestimate them during the off-peak hours. It is argued
that this may occur because startup costs are not included in the
model.

In this paper, we extend the model presented in \cite{troha2014calculation}
and include the startup costs. Various methodologies have already
been proposed on how to include the startup costs (see \cite{martinez2008amethodology},
\cite{gribik2007marketclearing} and \cite{zhang2009onreducing} for
example). Most of them rely on a price uplift approach, where first
the power price without startup costs is calculated. This price is
then uplifted to reflect the startup costs. In our model, the startup
costs are included in a mathematically rigorous fashion without relying
on the uplift heuristic. 

We show that startup costs are responsible for introducing many spikes
in spot electricity prices. To reduce the number of spikes, we include
the grid operator, who is responsible for managing the grid and for
a reliable delivery of electricity, by enhancing our model in the
second part of the paper.

This paper is organized as follows: In Section \ref{sec:Problem-description}
we give a detailed mathematical description of the model, and in Section
\ref{sec:N1} we present the numerical results. Numerical results
motivate us to introduce the grid operator in Section \ref{sec:Grid-operator}.
We conclude the paper in Section \ref{sec:Conclusions}.

\section{Problem description\label{sec:Problem-description}}

In this section we provide a detailed description of a model that
we use for the purpose of modeling the term structure of electricity
prices. The model belongs to a class of game theoretic equilibrium
models. Market participants are divided into consumers and producers.
A set of consumers is denoted by $C$ and has cardinality $0<\left|C\right|<\infty$.
Similarly, a set of producers is denoted by $P$ and has cardinality
$0<\left|P\right|<\infty$. Each producer owns a portfolio of power
plants that can have different characteristics such as capacity, startup
costs, ramp-up and ramp-down constraints, efficiency, and fuel type.
The set of all fuel types is denoted by $L$. Sets $R^{p,l}$ denote
all power plants owned by producer $p\in P$ that run on fuel $l\in L$.
A set $R^{p,l}$ may be empty since each producer typically does not
own all possible types of power plants. Moreover, this allows us to
include non physical traders such as banks or speculators, who do
not own any electricity generation facilities and are without a physical
demand for electricity, as producers $p\in P$ with $R^{p,l}=\left\{ \right\} $
for all $l\in L$. 

As we will see in Section \ref{sub:The-hypothetical-market}, it is
useful to introduce another player named the hypothetical market agent
besides producers and consumers. The hypothetical market agent plays
the role of the electricity market and ensures that the term structure
of the electricity price is such that the market clearing condition
is satisfied for all electricity forward contracts.

We are interested in delivery times $T_{j}$, $j\in J=\left\{ 1,...,T'\right\} $,
where power for each delivery time $T_{j}$ can be traded through
numerous forward contracts at times $t_{i}$, $i\in I_{j}$. The electricity
price at time $t_{i}$ for delivery at time $T_{j}$ is denoted by
$\Pi\left(t_{i},T_{j}\right)$. Since contracts with trading time
later than delivery time do not exist, we require $t_{\max\left\{ I_{j}\right\} }=T_{j}$
for all $j\in J$. The number of all forward contracts, i.e. $\sum_{j\in J}\left|I_{j}\right|$,
is denoted by $N$. Uncertainty is modeled by a filtered probability
space $\left(\Omega,\mathcal{F},\mathbb{F}=\left\{ \mathcal{F}_{t},t\in I\right\} ,\mathbb{P}\right)$,
where $I=\cup_{j\in J}I_{j}$. The $\sigma$-algebra $\mathcal{F}_{t}$
represents information available at time $t$.

The exogenous variables that appear in our model are (a) aggregate
power demand $D\left(T_{j}\right)$ for each delivery period $j\in J$,
(b) prices of fuel forward contracts $G_{l}\left(t_{i},T_{j}\right)$
for each fuel $l\in L$, delivery period $j\in J$, and trading period
$i\in I_{j}$, and (c) prices of emissions forward contracts $G_{em}\left(t_{i},T_{j}\right)$,
$j\in J$, $i\in I_{j}$. Electricity prices and all exogenous variables
are assumed to be adapted to the filtration $\left\{ \mathcal{F}_{t}\right\} _{t\in I}$
and have finite second moments.

Let $v_{k}\in\mathbb{R}^{n_{k}}$, $n_{k}\in\mathbb{N}$, $k\in K$,
and $K=\left\{ 1,...,\left|K\right|\right\} $ be given vectors. For
convenience, we define a vector concatenation operator as
\[
\left|\right|_{k\in K}v_{k}=\left[v_{1}^{\top},...,v_{\left|K\right|}^{\top}\right]^{\top}.
\]

\subsection{Producers\label{sub:Producer}}

Each producer $p\in P$ participates in the electricity, fuel, and
emission markets. Forward as well as spot contracts are available
on all markets. Electricity prices, fuel prices, and emission prices
are denoted by $\Pi\left(t_{i},T_{j}\right)$, $G_{l}\left(t_{i},T_{j}\right)$
where $l\in L$, and $G_{em}\left(t_{i},T_{j}\right)$, respectively. 

A producer may participate in the market by buying and selling forward
and spot contracts. The number of electricity forward contracts that
producer $p\in P$ buys at trading time $t_{i}$, $i\in I_{j}$ for
delivery at time $T_{j}$, $j\in J$ is denoted by $V_{p}\left(t_{i},T_{j}\right)$.
Similarly, the number of fuel and emission forward contracts that
producer $p\in P$ buys at trading time $t_{i}$, $i\in I_{j}$ for
delivery at time $T_{j}$, $j\in J$ is denoted by $F_{p,l}\left(t_{i},T_{j}\right)$,
$l\in L$ and $O_{p}\left(t_{i},T_{j}\right)$, respectively. Producers
own a generally non-empty portfolio of power plants. The actual production
of electricity from power plant $r\in R^{p,l}$ at delivery time $T_{j}$,
$j\in J$ is denoted by $\widehat{W}_{p,l,r}\left(T_{j}\right)$.

\subsubsection{Production variables}

In this section we investigate the production of power plants more
closely. Each power plant $r\in R^{p,l}$, $p\in P$, $l\in L$ has
a maximum export limit and minimum stable limit denoted by $\overline{W}_{max}^{p,l,r}\left(T_{j}\right)$
and $\overline{W}_{min}^{p,l,r}\left(T_{j}\right)$, respectively.
The maximum export limit defines the maximum production capacity of
a power plant and the minimum stable limit defines the minimum production
that a power plant is able to maintain for a longer period of time.
We allow each of the parameters to be time dependent to account for
the maintenance of power plants. 

Stable production of each power plant must satisfy
\begin{equation}
\widehat{W}_{p,l,r}\left(T_{j}\right)\in\left\{ 0\right\} \cup\left[\overline{W}_{min}^{p,l,r}\left(T_{j}\right),\overline{W}_{max}^{p,l,r}\left(T_{j}\right)\right]\label{eq:capacity-1}
\end{equation}
for each $j\in J$. It is allowed for a power plant to have production
$\widehat{W}_{p,l,r}\left(T_{j}\right)\in\left(0,\overline{W}_{min}^{p,l,r}\left(T_{j}\right)\right)$
for a very short period of time (i.e. during a ramp-up and ramp-down
phase). To formulate these constraints in an optimization framework,
we introduce new decision variables $W_{p,l,r}^{\left(k\right)}\left(T_{j}\right)$,
$k\in\left\{ 1,...,6\right\} $ with the following meaning:
\begin{itemize}
\item $W_{p,l,r}^{\left(1\right)}\left(T_{j}\right)$, $j\in J$ is a continuous
variable that is $1$ if the power plant is fully ramped up at time
$T_{j}$ and $0$ if the power plant is not producing at all at time
$T_{j}$. If $W_{p,l,r}^{\left(1\right)}\left(T_{j}\right)\in\left(0,1\right)$
then the power plant is in the ramp-up or ramp-down phase. In an optimization
framework, $W_{p,l,r}^{\left(1\right)}\left(T_{j}\right)$ is defined
as
\begin{equation}
W_{p,l,r}^{\left(1\right)}\left(T_{j}\right)\in\left[0,1\right].\label{eq:w1}
\end{equation}

\item $W_{p,l,r}^{\left(2\right)}\left(T_{j}\right)$, $j\in J$ is a binary
variable that is $1$ if the power plant is fully ramped up at time
$T_{j}$ and $0$ otherwise. In an optimization framework, $W_{p,l,r}^{\left(2\right)}\left(T_{j}\right)$
is defined as
\begin{equation}
\begin{array}{c}
W_{p,l,r}^{\left(2\right)}\left(T_{j}\right)\leq W_{p,l,r}^{\left(1\right)}\left(T_{j}\right)\\
\\
W_{p,l,r}^{\left(2\right)}\left(T_{j}\right)\in\left[0,1\right]
\end{array}\label{eq:w2_1}
\end{equation}
and 
\begin{equation}
W_{p,l,r}^{\left(2\right)}\left(T_{j}\right)\in\mathbb{Z}.\label{eq:w2_2_b}
\end{equation}

\item $W_{p,l,r}^{\left(3\right)}\left(T_{j}\right)$, $j\in J\backslash\left\{ 1\right\} $
is a continuous variable that denotes the increase of $W_{p,l,r}^{\left(1\right)}\left(T_{j}\right)$
from time $T_{j-1}$ to time $T_{j}$. In an optimization framework,
$W_{p,l,r}^{\left(3\right)}\left(T_{j}\right)$ is defined as
\begin{equation}
W_{p,l,r}^{\left(3\right)}\left(T_{j}\right)\geq W_{p,l,r}^{\left(1\right)}\left(T_{j}\right)-W_{p,l,r}^{\left(1\right)}\left(T_{j-1}\right)\label{eq:w3_1}
\end{equation}
and 
\begin{equation}
W_{p,l,r}^{\left(3\right)}\left(T_{j}\right)\in\left[0,1\right].\label{eq:w3_2}
\end{equation}
 
\item $W_{p,l,r}^{\left(4\right)}\left(T_{j}\right)$, $j\in J\backslash\left\{ 1\right\} $
is a binary variable that is $1$ if the power plant is in the ramp-up
phase and $0$ otherwise. In an optimization framework, $W_{p,l,r}^{\left(4\right)}\left(T_{j}\right)$
is defined as
\begin{equation}
\begin{array}{c}
W_{p,l,r}^{\left(4\right)}\left(T_{j}\right)\geq W_{p,l,r}^{\left(1\right)}\left(T_{j}\right)-W_{p,l,r}^{\left(1\right)}\left(T_{j-1}\right)\\
\\
W_{p,l,r}^{\left(4\right)}\left(T_{j}\right)\in\left[0,1\right]
\end{array}\label{eq:w4_1}
\end{equation}
and
\begin{equation}
W_{p,l,r}^{\left(4\right)}\left(T_{j}\right)\in\mathbb{Z}.\label{eq:w4_2_b}
\end{equation}

\item $W_{p,l,r}^{\left(5\right)}\left(T_{j}\right)$, $j\in J\backslash\left\{ 1\right\} $
is a binary variable that is $1$ if the power plant is in the ramp-down
phase and $0$ otherwise. In an optimization framework, $W_{p,l,r}^{\left(5\right)}\left(T_{j}\right)$
is defined as
\begin{equation}
\begin{array}{c}
W_{p,l,r}^{\left(5\right)}\left(T_{j}\right)\geq W_{p,l,r}^{\left(1\right)}\left(T_{j-1}\right)-W_{p,l,r}^{\left(1\right)}\left(T_{j}\right)\\
\\
W_{p,l,r}^{\left(5\right)}\left(T_{j}\right)\in\left[0,1\right]
\end{array}\label{eq:w5_1}
\end{equation}
and
\begin{equation}
W_{p,l,r}^{\left(5\right)}\left(T_{j}\right)\in\mathbb{Z}.\label{eq:w5_2_b}
\end{equation}

\item $W_{p,l,r}^{\left(6\right)}\left(T_{j}\right)$, $j\in J$ is a continuous
variable such that 
\begin{equation}
\widehat{W}_{p,l,r}\left(T_{j}\right)=W_{p,l,r}^{\left(1\right)}\left(T_{j}\right)\overline{W}_{min}^{p,l,r}\left(T_{j}\right)+W_{p,l,r}^{\left(6\right)}\left(T_{j}\right)\left(\overline{W}_{max}^{p,l,r}\left(T_{j}\right)-\overline{W}_{min}^{p,l,r}\left(T_{j}\right)\right)\label{eq:prod_vol-1}
\end{equation}
where
\begin{equation}
W_{p,l,r}^{\left(6\right)}\left(T_{j}\right)\leq W_{p,l,r}^{\left(2\right)}\left(T_{j}\right)\label{eq:w6_1}
\end{equation}
and
\begin{equation}
W_{p,l,r}^{\left(6\right)}\left(T_{j}\right)\in\left[0,1\right].\label{eq:w6_2}
\end{equation}

\end{itemize}
Variable $W_{p,l,r}^{\left(1\right)}\left(T_{j}\right)$ tells us
whether the power plant is running at time $T_{j}$. If the power
plant is not running at time $T_{j}$, then by (\ref{eq:w2_1}) and
(\ref{eq:w6_1}), $W_{p,l,r}^{\left(6\right)}\left(T_{j}\right)=0$
and by (\ref{eq:prod_vol-1}) also $\widehat{W}_{p,l,r}\left(T_{j}\right)=0$.
On the other hand, if the power plant is fully ramped up time $T_{j}$,
then $W_{p,l,r}^{\left(1\right)}\left(T_{j}\right)=1$ and $W_{p,l,r}^{\left(6\right)}\left(T_{j}\right)\in\left[0,1\right]$,
and thus $\widehat{W}_{p,l,r}\left(T_{j}\right)\in\left[\overline{W}_{min}^{p,l,r}\left(T_{j}\right),\overline{W}_{max}^{p,l,r}\left(T_{j}\right)\right]$.

\subsubsection{Maximum ramp-up and maximum ramp-down constraints\label{sub:Maximum-ramp-up}}

Producer $p\in P$ is not able to arbitrarily choose her decision
variables because there are some constraints that limit her feasible
set. The change in production of each power plant from one delivery
period to next is limited by the ramp-up and ramp-down constraints.
For each $j\in\left\{ 1,...,T'-1\right\} $, where $T'$ denotes the
last delivery period, $l\in L$ and $r\in R^{p,l}$ these constraints
can be expressed as 
\begin{equation}
\triangle\overline{W}_{min}^{p,l,r}\left(T_{j}\right)\leq\widehat{W}_{p,l,r}\left(T_{j+1}\right)-\widehat{W}_{p,l,r}\left(T_{j}\right)\leq\triangle\overline{W}_{max}^{p,l,r}\left(T_{j}\right),\label{eq:pb0}
\end{equation}
where $\triangle\overline{W}_{max}^{p,l,r}$ and $\triangle\overline{W}_{min}^{p,l,r}$
represent maximum rates for ramping up and down, respectively. The
ramping rates highly depend on the type of the power plant. Some gas
power plants can increase production from zero to the maximum in just
a few minutes, while the same action may take days or weeks for a
nuclear power plant.

Using (\ref{eq:prod_vol-1}), we can rewrite Constraint (\ref{eq:pb0})
for all $j\in\left\{ 1,...,T'-1\right\} $ as 
\begin{equation}
\begin{array}{rcl}
\overline{W}_{min}^{p,l,r}\left(T_{j}\right) & \leq & W_{p,l,r}^{\left(1\right)}\left(T_{j+1}\right)\overline{W}_{min}^{p,l,r}\left(T_{j+1}\right)+W_{p,l,r}^{\left(6\right)}\left(T_{j+1}\right)\left(\overline{W}_{max}^{p,l,r}\left(T_{j+1}\right)-\overline{W}_{min}^{p,l,r}\left(T_{j+1}\right)\right)\\
\\
 &  & -W_{p,l,r}^{\left(1\right)}\left(T_{j}\right)\overline{W}_{min}^{p,l,r}\left(T_{j}\right)-W_{p,l,r}^{\left(6\right)}\left(T_{j}\right)\left(\overline{W}_{max}^{p,l,r}\left(T_{j}\right)-\overline{W}_{min}^{p,l,r}\left(T_{j}\right)\right)\\
\\
 & \leq & \triangle\overline{W}_{max}^{p,l,r}\left(T_{j}\right).
\end{array}\label{eq:ramp}
\end{equation}
Additionally, if the power plant is in a ramp-up phase, then it has
to increase production and finish the ramp-up phase as fast as possible.
Such a requirement can be enforced as
\begin{equation}
W_{p,l,r}^{\left(1\right)}\left(T_{j+1}\right)\geq\min\left\{ W_{p,l,r}^{\left(1\right)}\left(T_{j}\right)+\frac{\triangle\overline{W}_{max}^{p,l,r}\left(T_{j}\right)}{\overline{W}_{min}^{p,l,r}\left(T_{j}\right)},1\right\} 
\end{equation}
where $j\in\left\{ 1,...,T'-1\right\} $. Since this constraint is
relevant only during the ramp-up phase, we reformulate it for $j\in\left\{ 1,...,T'-1\right\} $
as
\begin{equation}
W_{p,l,r}^{\left(1\right)}\left(T_{j+1}\right)\geq\min\left\{ W_{p,l,r}^{\left(1\right)}\left(T_{j}\right)+\frac{\triangle\overline{W}_{max}^{p,l,r}\left(T_{j}\right)}{\overline{W}_{min}^{p,l,r}\left(T_{j}\right)},1\right\} -M_{1}\left(1-W_{p,l,r}^{\left(4\right)}\left(T_{j}\right)\right),\label{eq:rampup_fast}
\end{equation}
where $M_{1}\geq1+\frac{\triangle\overline{W}_{max}^{p,l,r}\left(T_{j}\right)}{\overline{W}_{min}^{p,l,r}\left(T_{j}\right)}$.
Most of the available optimization solvers are not able to handle
constraints that include min or max functions. Thus, we apply a well
established approach to handle logical constraints, and introduce
a new binary decision variable $W_{p,l,r}^{\left(7\right)}\left(T_{j}\right)$
as 
\begin{equation}
W_{p,l,r}^{\left(7\right)}\left(T_{j}\right)\in\left[0,1\right]\label{eq:bounds_w7}
\end{equation}
and
\begin{equation}
W_{p,l,r}^{\left(7\right)}\left(T_{j}\right)\in\mathbb{Z}\label{eq:b_7}
\end{equation}
where $j\in J$, that makes sure that at least one of the following
constraints
\begin{equation}
W_{p,l,r}^{\left(1\right)}\left(T_{j+1}\right)\geq W_{p,l,r}^{\left(1\right)}\left(T_{j}\right)+\frac{\triangle\overline{W}_{max}^{p,l,r}\left(T_{j}\right)}{\overline{W}_{min}^{p,l,r}\left(T_{j}\right)}-M_{1}\left(1-W_{p,l,r}^{\left(4\right)}\left(T_{j}\right)\right)-M_{2}W_{p,l,r}^{\left(7\right)}\left(T_{j}\right)\label{eq:ramp_up_eq1}
\end{equation}
and
\begin{equation}
W_{p,l,r}^{\left(1\right)}\left(T_{j+1}\right)\geq1-M_{1}\left(1-W_{p,l,r}^{\left(4\right)}\left(T_{j}\right)\right)-M_{2}\left(1-W_{p,l,r}^{\left(7\right)}\left(T_{j}\right)\right),\label{eq:ramp_up_eq2}
\end{equation}
where $M_{2}\geq1$, is enforced. 

Similarly, if a power plant is in the ramp-down phase, then it has
to decrease production and finish the ramp-down phase as fast as possible.
Such requirement can be enforced as
\begin{equation}
W_{p,l,r}^{\left(1\right)}\left(T_{j+1}\right)\leq\text{max}\left\{ W_{p,l,r}^{\left(1\right)}\left(T_{j}\right)-\frac{\triangle\overline{W}_{min}^{p,l,r}\left(T_{j}\right)}{\overline{W}_{min}^{p,l,r}\left(T_{j}\right)},0\right\} +M_{1}\left(1-W_{p,l,r}^{\left(5\right)}\left(T_{j}\right)\right)\label{eq:rampdown_fast}
\end{equation}
for $j\in\left\{ 1,...,T'-1\right\} $. Most of the available optimization
solvers are not able to handle constraints that include min or max
functions. We apply the approach described above and introduce a new
binary decision variable $W_{p,l,r}^{\left(8\right)}\left(T_{j}\right)$
as 
\begin{equation}
W_{p,l,r}^{\left(8\right)}\left(T_{j}\right)\in\left[0,1\right]\label{eq:bounds_w8}
\end{equation}
and
\begin{equation}
W_{p,l,r}^{\left(8\right)}\left(T_{j}\right)\in\mathbb{Z}\label{eq:b_8}
\end{equation}
where $j\in J$, that makes sure that at least one of the following
constraints
\begin{equation}
W_{p,l,r}^{\left(1\right)}\left(T_{j+1}\right)\leq W_{p,l,r}^{\left(1\right)}\left(T_{j}\right)-\frac{\triangle\overline{W}_{min}^{p,l,r}\left(T_{j}\right)}{\overline{W}_{min}^{p,l,r}\left(T_{j}\right)}+M_{1}\left(1-W_{p,l,r}^{\left(5\right)}\left(T_{j}\right)\right)+M_{2}W_{p,l,r}^{\left(8\right)}\left(T_{j}\right)\label{eq:ramp_down_eq1}
\end{equation}
and
\begin{equation}
W_{p,l,r}^{\left(1\right)}\left(T_{j+1}\right)\leq M_{1}\left(1-W_{p,l,r}^{\left(5\right)}\left(T_{j}\right)\right)+M_{2}\left(1-W_{p,l,r}^{\left(8\right)}\left(T_{j}\right)\right)\label{eq:ramp_down_eq2}
\end{equation}
is enforced.

\subsubsection{Other inequality constraints}

We bound the the number of electricity contracts that each producer
is allowed to trade as 
\begin{equation}
-V_{trade}\leq V_{p}\left(t_{i},T_{j}\right)\leq V_{trade}\label{eq:pb2}
\end{equation}
for some large $V_{trade}>0$. Trading of an infinite number of contracts
would clearly lead to a bankruptcy of one of the \foreignlanguage{british}{counterparties}
involved and must thus be prevented. In \cite{troha2014theexistence}
it was shown, that if $V_{trade}$ is chosen to be large enough, then
Constraint (\ref{eq:pb2}) has no impact on the optimal solution and
can be eliminated from the problem.

\subsubsection{Equality constraints}

There are also equality constraints that connect power plant production
with electricity, fuel, and emission trading. For each $j\in J$ the
electricity sold in the forward and spot market together must equal
the actually produced electricity, i.e.
\begin{equation}
-\sum_{i\in I_{j}}V_{p}\left(t_{i},T_{j}\right)=\sum_{l\in L}\sum_{r\in R^{p,l}}\widehat{W}_{p,l,r}\left(T_{j}\right).\label{eq:pb2-1}
\end{equation}
Each producer $p\in P$ has to make sure that a sufficient amount
of fuel $l\in L$ has been bought to cover the electricity production
for each delivery period $j\in J$ . Such constraint can be expressed
as 
\begin{equation}
\sum_{r\in R^{p,l}}\widehat{W}_{p,l,r}\left(T_{j}\right)c^{p,l,r}=\sum_{i\in I_{j}}F_{p,l}\left(t_{i},T_{j}\right)\label{eq:pb3-1}
\end{equation}
where $c^{p,l,r}>0$ is the efficiency of power plant $r\in R^{p,l}$. 

The carbon emission obligation constraint can be written as
\begin{equation}
\sum_{j\in J}\sum_{i\in I_{j}}O\left(t_{i},T_{j}\right)=\sum_{j\in J}\sum_{l\in L}\sum_{r\in R^{p,l}}\widehat{W}_{p,l,r}\left(T_{j}\right)g^{p,l,r},\label{eq:pb4}
\end{equation}
where $g^{p,l,r}>0$ denotes the carbon emission intensity factor
for power plant $r\in R^{p,l}$. This constraint ensures that enough
emission certificates have been bought to cover the electricity production
over the whole planning horizon.

\subsubsection{Producers' optimization problem}

The notation of the decision variables is greatly simplified if they
are concatenated into 
\begin{itemize}
\item electricity trading vectors $V_{p}\left(T_{j}\right)=\left|\right|_{i\in I_{j}}V_{p}\left(t_{i},T_{j}\right)$
and $V_{p}=\left|\right|_{j\in J}V_{p}\left(T_{j}\right)$, 
\item fuel trading vectors $F_{p}\left(t_{i},T_{j}\right)=\left|\right|_{l\in L}F_{p,l}\left(t_{i},T_{j}\right)$,
$F_{p}\left(T_{j}\right)=\left|\right|_{i\in I_{j}}F_{p}\left(t_{i},T_{j}\right)$,
and\linebreak{}
 $F_{p}=\left|\right|_{j\in J}F_{p}\left(T_{j}\right)$, 
\item emission trading vectors $O_{p}\left(T_{j}\right)=\left|\right|_{i\in I_{j}}O_{p}\left(t_{i},T_{j}\right)$
and $O_{p}=\left|\right|_{j\in J}O_{p}\left(T_{j}\right)$, 
\item electricity production vectors $W_{p,l,r}\left(T_{j}\right)=\left|\right|_{k\in\left\{ 1,..,8\right\} }W_{p,l,r}^{\left(k\right)}\left(T_{j}\right)$,\linebreak{}
 $W_{p,l}\left(T_{j}\right)=\left|\right|_{r\in R^{p,l}}W_{p,l,r}\left(T_{j}\right)$,
$W_{p}\left(T_{j}\right)=\left|\right|_{l\in L}W_{p,l}\left(T_{j}\right)$,
and $W_{p}=\left|\right|_{j\in J}W_{p}\left(T_{j}\right)$, 
\end{itemize}
and finally $v_{p}=\left[V_{p}^{\top},F_{p}^{\top},O_{p}^{\top},W_{p}^{\top}\right]^{\top}$.

Similarly, the notation of the prices is greatly simplified if they
are concatenated into 
\begin{itemize}
\item electricity price vectors $\Pi\left(T_{j}\right)=\left|\right|_{i\in I_{j}}\Pi\left(t_{i},T_{j}\right)$,
and $\Pi=\left|\right|_{j\in J}e^{-\hat{r}T_{j}}\Pi\left(T_{j}\right)$,
where $\hat{r}\in\mathbb{R}$ is a constant interest rate, 
\item fuel price vectors $G\left(t_{i},T_{j}\right)=\left|\right|_{l\in L}G_{l}\left(t_{i},T_{j}\right)$,
$G\left(T_{j}\right)=\left|\right|_{i\in I_{j}}G\left(t_{i},T_{j}\right)$,
and\linebreak{}
$G=\left|\right|_{j\in J}e^{-\hat{r}T_{j}}G\left(T_{j}\right)$, 
\item emission price vector $G_{em}\left(T_{j}\right)=\left|\right|_{i\in I_{j}}G_{em}\left(t_{i},T_{j}\right)$,
and $G_{em}=\left|\right|_{j\in J}e^{-\hat{r}T_{j}}G_{em}\left(T_{j}\right)$, 
\item startup costs vector $\widehat{s}^{p,l,r}=\left[0,0,s^{p,l,r},0,0,0,0,0\right]^{\top}$,
$s^{p,l}=\left|\right|_{r\in R^{p,l}}\widehat{s}^{p,l,r}$, $s^{p}=\left|\right|_{l\in L}s^{p,l}$,
and $\widehat{s}^{p}=\left|\right|_{j\in J}e^{-\hat{r}T_{j}}s^{p}$,
where $s^{p,l,r}\geq0$ denotes the startup costs of power plant $r\in R^{p,l}$,
\end{itemize}
and finally 
\[
\pi_{p}=\left[\Pi^{\top},G^{\top},G_{em}^{\top},\left(\widehat{s}^{p}\right)^{\top}\right]^{\top}.
\]

Any producers' goal is to maximize their expected profit subject to
a risk budget. In this work we assume that the risk budget is expressed
in a mean-variance framework. The main argument that supports this
decision is that delta hedging, which is the most widely used hedging
strategy, can be captured in this framework. 

The profit $P_{p}\left(v_{p},\pi_{p}\right)$ of producer $p\in P$
can be calculated as
\begin{equation}
P_{p}\left(v_{p},\pi_{p}\right)=\sum_{j\in J}e^{-\hat{r}T_{j}}\left(\sum_{i\in I_{j}}P_{p}^{t_{i},T_{j}}\left(v_{p},\pi_{p}\right)-\sum_{l\in L}\sum_{r\in R^{p,l}}s^{p,l,r}W_{p,l,r}^{\left(3\right)}\left(T_{j}\right)\right),
\end{equation}
 where the profit $P_{p}^{t_{i},T_{j}}\left(v_{p},\pi_{p}\right)$
for each $i\in I_{j}$ and $j\in J$ can be calculated as

\[
P_{p}^{t_{i},T_{j}}\left(v_{p},\pi_{p}\right)=-\Pi\left(t_{i},T_{j}\right)V_{p}\left(t_{i},T_{j}\right)-O_{p}\left(t_{i},T_{j}\right)G_{em}\left(t_{i},T_{j}\right)-\sum_{l\in L}G_{l}\left(t_{i},T_{j}\right)F_{p,l}\left(t_{i},T_{j}\right).
\]
Under a mean-variance optimization framework, producers are interested
in the mean-variance utility

\[
\begin{array}{rcl}
\Psi_{p}\left(v_{p}\right) & = & \mathbb{E}^{\mathbb{P}}\left[P_{p}\left(v_{p},\pi_{p}\right)\right]-\frac{\lambda_{p}}{2}\text{Var}^{\mathbb{P}}\left[P_{p}\left(v_{p},\pi_{p}\right)\right]\\
\\
 & = & -\mathbb{E}^{\mathbb{P}}\left[\pi_{p}\right]^{\top}v_{p}-\frac{1}{2}\lambda_{p}v_{p}^{\top}Q_{p}v_{p},
\end{array}
\]
where $\lambda_{p}>0$ is their risk preference parameter and $Q_{p}:=\mathbb{E}^{\mathbb{P}}\left[\left(\pi_{p}-\mathbb{E}^{\mathbb{P}}\left[\pi_{p}\right]\right)\left(\pi_{p}-\mathbb{E}^{\mathbb{P}}\left[\pi_{p}\right]\right)^{\top}\right]$
an ``extended'' covariance matrix. Their objective is to solve the
following optimization problem
\begin{equation}
\begin{array}{rl}
\Phi_{p}=\underset{v_{p}}{\text{max }} & \Psi_{p}\left(v_{p}\right)\end{array}\tag{PR}\label{eq:opt_prod}
\end{equation}
subject to (\ref{eq:w1}), (\ref{eq:w2_1}), (\ref{eq:w2_2_b}), (\ref{eq:w3_1}),
(\ref{eq:w3_2}), (\ref{eq:w4_1}), (\ref{eq:w4_2_b}), (\ref{eq:w5_1}),
(\ref{eq:w5_2_b}), (\ref{eq:w6_1}), (\ref{eq:w6_2}), (\ref{eq:ramp}),
(\ref{eq:bounds_w7}), (\ref{eq:b_7}), (\ref{eq:ramp_up_eq1}), (\ref{eq:ramp_up_eq2}),
(\ref{eq:bounds_w8}), (\ref{eq:b_8}), (\ref{eq:ramp_down_eq1}),
(\ref{eq:ramp_down_eq2}), (\ref{eq:pb2}), (\ref{eq:pb2-1}), (\ref{eq:pb3-1}),
and (\ref{eq:pb4}). 

A standard approach to solving optimization problem with binary constraints
is to consider its continuous relaxation. We define a continuous relaxation
of Problem (\ref{eq:opt_prod}) as 
\begin{equation}
\begin{array}{rl}
\Phi_{p}=\underset{v_{p}}{\text{max }} & \Psi_{p}\left(v_{p}\right)\end{array}\tag{{\ensuremath{\widetilde{PR}}}}\label{eq:opt_prod_rel}
\end{equation}
subject to (\ref{eq:w1}), (\ref{eq:w2_1}), (\ref{eq:w3_1}), (\ref{eq:w3_2}),
(\ref{eq:w4_1}), (\ref{eq:w5_1}), (\ref{eq:w6_1}), (\ref{eq:w6_2}),
(\ref{eq:ramp}), (\ref{eq:bounds_w7}), (\ref{eq:ramp_up_eq1}),
(\ref{eq:ramp_up_eq2}), (\ref{eq:ramp_down_eq1}), (\ref{eq:ramp_down_eq2}),
(\ref{eq:bounds_w8}), (\ref{eq:pb2}), (\ref{eq:pb2-1}), (\ref{eq:pb3-1}),
and (\ref{eq:pb4}). Problem \ref{eq:opt_prod_rel} is the same as
problem Problem (\ref{eq:opt_prod}) except that it does not include
integrality constraints (\ref{eq:w2_2_b}), (\ref{eq:w4_2_b}), (\ref{eq:w5_2_b}),
(\ref{eq:b_7}) and (\ref{eq:b_8}).

\subsection{Consumers}

We make the assumption that demand is completely inelastic and that
each consumer $c\in C$ is responsible for satisfying a proportion
$p_{c}\in\left[0,1\right]$ of the total demand $D\left(T_{j}\right)$
at time $T_{j}$, $j\in J$. Since $p_{c}$ is a proportion, we clearly
have that $\sum_{c\in C}p_{c}=1.$ 

A number of electricity forward contracts consumer $c\in C$ buys
at trading time $t_{i}$, $i\in I_{j}$ for delivery at time $T_{j}$,
$j\in J$ is denoted by $V_{c}\left(t_{i},T_{j}\right)$.

\subsubsection{Inequality constraints}

We bound the the number of electricity contracts that each consumer
is allowed to trade as 
\begin{equation}
-V_{trade}\leq V_{p}\left(t_{i},T_{j}\right)\leq V_{trade}\label{eq:pb2-2}
\end{equation}
for some large $V_{trade}>0$. Trading of an infinite number of contracts
would clearly lead to a bankruptcy of one of the counterparties involved
and must thus be prevented. In \cite{troha2014theexistence} it was
shown, that if $V_{trade}$ is chosen large enough, then Constraint
(\ref{eq:pb2-2}) has no impact on the optimal solution and can be
eliminated from the problem.

\subsubsection{Equality constraints}

Consumers are responsible for satisfying the electricity demand of
end users. The electricity demand is expected to be satisfied for
each $T_{j}$, i.e.
\begin{equation}
\sum_{i\in I_{j}}V_{c}\left(t_{i},T_{j}\right)=p_{c}D\left(T_{j}\right).\label{eq:cb2}
\end{equation}
At the time of calculating the optimal decisions, consumers assume
that they know the future realization of demand $D\left(T_{j}\right)$
precisely. If the knowledge about the future realization of the demand
changes, then players can take recourse actions by recalculating their
optimal decisions with the updated demand forecast. Consumers may
assume that they will be able to execute the recourse actions, because
it is the job of the grid operator to ensure that a sufficient amount
of electricity is available on the market.

\subsubsection{Consumers' optimization problem}

Similarly as for producers, we can simplify the notation by introducing
electricity trading vectors $V_{c}\left(T_{j}\right)=\left|\right|_{i\in I_{j}}V_{c}\left(t_{i},T_{j}\right)$
and $V_{c}=\left|\right|_{j\in J}V_{c}\left(T_{j}\right)$.

Consumers would like to maximize their profit subject to a risk budget.
Similar to the model we introduced for producers, we assume that the
risk budget can be expressed in a mean-variance framework. The profit
of consumer $c\in C$ can be calculated as
\begin{equation}
P_{c}\left(V_{c},\Pi\right)=\sum_{j\in J}e^{-\hat{r}T_{j}}\left(\sum_{i\in I_{j}}-\Pi\left(t_{i},T_{j}\right)V_{c}\left(t_{i},T_{j}\right)+s_{c}p_{c}D\left(T_{j}\right)\right),\label{eq:23}
\end{equation}
where $\hat{r}\in\mathbb{R}$ denotes a constant interest rate and
$s_{c}\in\mathbb{R}$ denotes a contractually fixed price that consumer
$c\in C$ receives for selling the electricity further to end users
(e.g. households, businesses etc.). Note that the contractually fixed
price $s_{c}$ only affects the optimal objective value of consumer
$c\in C$, but not also her optimal solution. Since we are primarily
interested in optimal solutions, we simplify the notation and set
$s_{c}=0$. The correct optimal value can always be calculated via
post-processing when an optimal solution is already known. This may
be needed for risk management purposes. Note that in reality, end
users can change their electricity providers and consequently the
proportions $p_{c}$, $c\in C$. One could model the end user electricity
market with a similar equilibrium model as presented here, but this
is not the focus of this paper. Here we assume that proportions $p_{c}$
are constant for the period of our interest.

Under a mean-variance optimization framework consumers are interested
in the mean-variance utility

\[
\begin{array}{rcl}
\Psi_{c}\left(V_{c}\right) & = & \mathbb{E}^{\mathbb{P}}\left[P_{c}\left(V_{c},\Pi\right)\right]-\frac{\lambda_{c}}{2}\text{Var}^{\mathbb{P}}\left[P_{c}\left(V_{c},\Pi\right)\right]\\
\\
 & = & -\mathbb{E}^{\mathbb{P}}\left[\Pi\right]^{\top}V_{c}-\frac{\lambda_{c}}{2}V_{c}^{\top}Q_{c}V_{c},
\end{array}
\]
where $\lambda_{c}>0$ is their risk preference and $Q_{c}:=\mathbb{E}^{\mathbb{P}}\left[\left(\Pi-\mathbb{E}^{\mathbb{P}}\left[\Pi\right]\right)\left(\Pi-\mathbb{E}^{\mathbb{P}}\left[\Pi\right]\right)^{\top}\right]$
a covariance matrix. Their objective is to solve the following optimization
problem
\begin{equation}
\Phi_{c}=\underset{V_{c}}{\text{max }}\Psi_{c}\left(V_{c}\right)\tag{CO}\label{eq:opt_con}
\end{equation}
subject to (\ref{eq:pb2-2}) and (\ref{eq:cb2}).

\subsection{Matrix notation}

The analysis of the problem is greatly simplified if a more compact
notation is introduced.

Equality constraints of producer $p\in P$ can be expressed as
\[
A_{p}v_{p}=0
\]
and inequality constraints as 
\[
B_{p}v_{p}\leq b_{p}
\]
for some $A_{p}\in\mathbb{R}^{\left|J\right|\left(\left|L\right|+1\right)+1\times\dim v_{p}}$,
$B_{p}\in\mathbb{R}^{n_{p}\times\dim v_{p}}$ and $b_{p}\in\mathbb{R}^{n_{p}}$,
where $n_{p}$ denotes the number of the inequality constraints of
producer $p\in P$. Define feasible sets 
\[
\widetilde{S}_{p}:=\left\{ v_{p}:A_{p}v_{p}=a_{p}\text{ \text{and }}B_{p}v_{p}\leq b_{p}\right\} 
\]
and
\[
S_{p}:=\left\{ v_{p}:A_{p}v_{p}=a_{p}\text{ \text{and }}B_{p}v_{p}\leq b_{p}\text{ \text{and }}\left[v_{p}\right]_{i}\in\left\{ 0,1\right\} \:\forall i\in\mathcal{I}\right\} ,
\]
where $\mathcal{I}$ denotes a set of decisions variables with binarity
constraints (i.e. $W_{p,l,r}^{\left(k\right)}\left(T_{j}\right)$
for all $k\in\left\{ 2,4,5,7,8\right\} $, $r\in R^{p,l}$ and for
all $j\in J$).

It is useful to investigate the inner structure of the matrices. By
considering equality constraints (\ref{eq:pb2-1}), (\ref{eq:pb3-1}),
and (\ref{eq:pb4}) we can see that 
\begin{equation}
A_{p}=\left[\begin{array}{ccc}
\hat{A}_{1} & 0 & \hat{A}_{3,p}\\
0 & \hat{A}_{2} & \hat{A}_{4,p}
\end{array}\right]\label{eq:Ap}
\end{equation}
where $\hat{A}_{1}\in\mathbb{R}^{\left|J\right|\times N},\hat{A}_{2}\in\mathbb{R}^{\left(\left|J\right|\left|L\right|+1\right)\times N\left(\left|L\right|+1\right)},\hat{A}_{3,p}\in\mathbb{R}^{\left|J\right|\times\dim W_{p}},\hat{A}_{4,p}\in\mathbb{R}^{\left(\left|J\right|\left|L\right|+1\right)\times\dim W_{p}}$.
One can see that matrices $\hat{A}_{1}$ and $\hat{A}_{2}$ are independent
of producer $p\in P$ and matrices $\hat{A}_{3,p}$ and $\hat{A}_{4,p}$
depend on producer $p\in P$. One can further investigate the structure
of $\hat{A}_{1}$ and see 
\begin{equation}
\hat{A}_{1}=\left[\begin{array}{ccc}
1_{1} &  & 0\\
 & \ddots\\
0 &  & 1_{\left|J\right|}
\end{array}\right],\label{eq:A1}
\end{equation}
where $1_{j}$, $j\in J$ is a row vector of ones of length $\left|I_{j}\right|$.
Similarly,
\begin{equation}
\hat{A}_{2}=\left[\begin{array}{cccc}
\hat{A}_{1} & \cdots & 0 & 0\\
\vdots & \ddots & \vdots & \vdots\\
0 & \cdots & \hat{A}_{1} & 0\\
0 & \cdots & 0 & 1_{N}
\end{array}\right],\label{eq:A2}
\end{equation}
where the number of rows in the block notation above is $\left|L\right|+1$.
The first $\left|L\right|$ rows correspond to (\ref{eq:pb3-1}) and
the last row corresponds to (\ref{eq:pb4}). 

The profit of producer $p\in P$ can be written as
\[
P_{p}\left(v_{p},\pi_{p}\right)=-\pi_{p}^{\top}v_{p}.
\]
In a compact notation, the mean-variance utility of producer $p\in P$
can be calculated as
\[
\begin{array}{rcl}
\Psi_{p}\left(v_{p},\mathbb{E}^{\mathbb{P}}\left[\Pi\right]\right) & = & \mathbb{E}^{\mathbb{P}}\left[-\pi_{p}^{\top}v_{p}-\frac{1}{2}\lambda_{p}v_{p}^{\top}\left(\pi_{p}-\mathbb{E}^{\mathbb{P}}\left[\pi_{p}\right]\right)\left(\pi_{p}-\mathbb{E}^{\mathbb{P}}\left[\pi_{p}\right]\right)^{\top}v_{p}\right]\\
\\
 & = & -\mathbb{E}^{\mathbb{P}}\left[\pi_{p}\right]^{\top}v_{p}-\frac{1}{2}\lambda_{p}v_{p}^{\top}Q_{p}v_{p},
\end{array}
\]
where
\begin{equation}
Q_{p}:=\mathbb{E}^{\mathbb{P}}\left[\left(\pi_{p}-\mathbb{E}^{\mathbb{P}}\left[\pi_{p}\right]\right)\left(\pi_{p}-\mathbb{E}^{\mathbb{P}}\left[\pi_{p}\right]\right)^{\top}\right].
\end{equation}
The inner structure of matrix $Q_{p}$ is the following 
\begin{equation}
Q_{p}=\left[\begin{array}{ccc}
\hat{Q}_{1} & \hat{Q}_{2} & 0\\
\hat{Q}_{2}^{\top} & \hat{Q}_{3} & 0\\
0 & 0 & 0
\end{array}\right]
\end{equation}
where $\hat{Q}_{1}\in\mathbb{R}^{N\times N},\hat{Q}_{2}\in\mathbb{R}^{N\times\left(\dim B_{p}+\dim O_{p}\right)}=\mathbb{R}^{N\times N\left(\left|L\right|+1\right)},\hat{Q}_{3}\in\mathbb{R}^{N\left(\left|L\right|+1\right)\times N\left(\left|L\right|+1\right)}$.
One can see that $\hat{Q}_{1}$, $\hat{Q}_{2}$, and $\hat{Q}_{3}$
do not depend on producer $p\in P$. The size of the larger matrix
$Q_{p}$ depends on producer $p\in P$, because different producers
have different number of power plants. 

Producer $p\in P$ attempts to solve the following optimization problem
\[
\Phi_{p}\left(\mathbb{E}^{\mathbb{P}}\left[\Pi\right]\right)=\underset{v_{p}\in S_{p}}{\text{max }}-\mathbb{E}^{\mathbb{P}}\left[\pi_{p}\right]^{\top}v_{p}-\frac{1}{2}\lambda_{p}v_{p}^{\top}Q_{p}v_{p},
\]
with the following continuous relaxation
\[
\Phi_{p}\left(\mathbb{E}^{\mathbb{P}}\left[\Pi\right]\right)=\underset{v_{p}\in\widetilde{S}_{p}}{\text{max }}-\mathbb{E}^{\mathbb{P}}\left[\pi_{p}\right]^{\top}v_{p}-\frac{1}{2}\lambda_{p}v_{p}^{\top}Q_{p}v_{p}.
\]

The equality constraints of consumer $c\in C$ can be expressed as
\[
A_{c}V_{c}=a_{c}
\]
and the inequality constraints as 
\[
B_{c}V_{c}\leq b_{c}
\]
where $A_{c}=\hat{A}_{1}$, $B_{c}\in\mathbb{R}^{2N\times N}$, $a_{c}\in\mathbb{R}^{\left|J\right|}$
and $b_{c}\in\mathbb{R}^{2N}$. Define a feasible set
\[
S_{c}:=\left\{ V_{c}\in\mathbb{R}^{N}:A_{c}V_{c}=a_{c}\text{ \text{and }}B_{c}V_{c}\leq b_{c}\right\} .
\]

The profit of consumer $c\in C$ can be written as
\[
P_{c}\left(V_{c},\Pi\right)=-\Pi^{\top}V_{c}.
\]
In a compact notation, the mean-variance utility of a consumer $c\in C$
can be calculated as
\[
\begin{array}{rcl}
\Psi_{c}\left(V_{c},\mathbb{E}^{\mathbb{P}}\left[\Pi\right]\right) & = & \mathbb{E}^{\mathbb{P}}\left[-\Pi^{\top}V_{c}-\frac{1}{2}\lambda_{c}V_{c}^{\top}\left(\Pi-\mathbb{E}^{\mathbb{P}}\left[\Pi\right]\right)\left(\Pi-\mathbb{E}^{\mathbb{P}}\left[\Pi\right]\right)^{\top}V_{c}\right]\\
\\
 & = & -\mathbb{E}^{\mathbb{P}}\left[\Pi\right]^{\top}V_{c}-\frac{\lambda_{c}}{2}V_{c}^{\top}Q_{c}V_{c},
\end{array}
\]
where
\begin{equation}
Q_{c}:=\mathbb{E}^{\mathbb{P}}\left[\left(\Pi-\mathbb{E}^{\mathbb{P}}\left[\Pi\right]\right)\left(\Pi-\mathbb{E}^{\mathbb{P}}\left[\Pi\right]\right)^{\top}\right].
\end{equation}
 Moreover, note that $Q_{c}=\hat{Q}_{1}$ for all $c\in C$. We set
$s_{c}=0$, w.l.o.g. Consumer $c\in C$ attempts to solve the following
optimization problem
\[
\Phi_{c}\left(\mathbb{E}^{\mathbb{P}}\left[\Pi\right]\right)=\underset{V_{c}\in S_{c}}{\text{max }}-\mathbb{E}^{\mathbb{P}}\left[\Pi\right]^{\top}V_{c}-\frac{\lambda_{c}}{2}V_{c}^{\top}Q_{c}V_{c}.
\]

\subsection{The hypothetical market agent\label{sub:The-hypothetical-market}}

Given the price vectors of electricity $\Pi$, fuel $G$, and emissions
$G_{em}$, each producer $p\in P$ and each consumer $c\in C$ can
calculate their optimal electricity trading vectors $V_{p}$ and $V_{c}$
by solving (\ref{eq:opt_prod}) and (\ref{eq:opt_con}), respectively.
However, the players are not necessary able to execute their calculated
optimal trading strategies because they may not find the counterparty
to trade with. In reality each contract consists of a buyer and a
seller, which imposes an additional constraint (also called the market
clearing constraint) that matches the number of short and long electricity
contracts for each $i\in I_{j}$ and $j\in J$ as follows,
\begin{equation}
\sum_{c\in C}V_{c}\left(t_{i},T_{j}\right)+\sum_{p\in P}V_{p}\left(t_{i},T_{j}\right)=0.\label{eq:mk}
\end{equation}
The electricity market is responsible for satisfying this constraint
by matching buyers with sellers. The matching is done through sharing
of the price and order book information among all market participants.
If at the current price there are more long contract than short contracts,
it means that the current price is too low and asks will start to
be submitted at higher prices. The converse occurs, if there are more
short contracts than long contracts. Eventually, the electricity price
at which the number of long and short contracts matches is found.
At such a price the constraint (\ref{eq:mk}) is satisfied ``naturally''
without explicitly requiring the players to satisfy it. They do so
because it is in their best interest, i.e. it maximizes their mean-variance
objective functions.

The question is how to formulate such an equilibrium constraint in
an optimization framework. A naive approach of writing the market
clearing constraint as an ordinary constraint forces the players to
satisfy it regardless of the price. We need a mechanism that models
the matching of buyers and sellers as it is performed by the electricity
market. For this purpose, we introduce a hypothetical market agent
who is allowed to slowly change electricity prices to ensure that
(\ref{eq:mk}) is satisfied.

Let the hypothetical market agent have the following profit function
\begin{equation}
\begin{array}{rcl}
P_{M}\left(\Pi,V\right) & = & \sum_{j\in J}e^{-\hat{r}T_{j}}\left[\sum_{i\in I_{j}}\Pi\left(t_{i},T_{j}\right)\left(\sum_{c\in C}V_{c}\left(t_{i},T_{j}\right)+\sum_{p\in P}V_{p}\left(t_{i},T_{j}\right)\right)\right]\\
\\
 & = & \mathbb{E}^{\mathbb{P}}\left[\Pi\right]^{\top}\left(\sum_{c\in C}V_{c}+\sum_{p\in P}V_{p}\right)
\end{array}\label{eq:market-1}
\end{equation}
and the expected profit
\begin{equation}
\Psi_{M}\left(\mathbb{E}^{\mathbb{P}}\left[\Pi\right],V\right)=\mathbb{E}^{\mathbb{P}}\left[P_{M}\left(V,\Pi\right)\right],\label{eq:utility_market}
\end{equation}
where $V=\left[V_{P}^{\top},V_{C}^{\top}\right]^{\top}$, $V_{P}=\left|\right|_{p\in P}V_{p}$,
and $V_{C}=\left|\right|_{c\in C}V_{c}$ and let the hypothetical
market agent attempts to solve
\begin{equation}
\Phi_{M}\left(V\right)=\underset{\mathbb{E}^{\mathbb{P}}\left[\Pi\right]}{\text{max }}\Psi_{M}\left(\mathbb{E}^{\mathbb{P}}\left[\Pi\right],V\right).\label{eq:SO}
\end{equation}
The KKT conditions for (\ref{eq:SO}) in the matrix notation read
\begin{equation}
\sum_{c\in C}V_{c}+\sum_{p\in P}V_{p}=0,\label{eq:KKT_hyp_mark}
\end{equation}
which is exactly the same as (\ref{eq:mk}). Note, that the equivalence
of (\ref{eq:mk}) and (\ref{eq:SO}) is a theoretical result that
has to be applied with caution in an algorithmic framework. Formulation
(\ref{eq:SO}) is clearly unstable since only a small mismatch in
the market clearing constraint sends the prices to $\pm\infty$. Thus,
a stable formulation of the hypothetical market agent must be found.
Let us now analyze the hypothetical market agent with the following,
slightly altered, optimization problem
\begin{equation}
\begin{array}{rl}
\underset{\mathbb{E}^{\mathbb{P}}\left[\Pi\right]}{\text{max }} & \Psi_{M}\left(\mathbb{E}^{\mathbb{P}}\left[\Pi\right],V\right)\\
\\
\text{s.t.} & \sum_{c\in C}V_{c}+\sum_{p\in P}V_{p}=0\\
\\
 & \mu_{M}=0,
\end{array}\tag{HMA}\label{eq:market_new}
\end{equation}
where $\mu_{M}$ denotes the dual variables of the equality constraint
in (\ref{eq:market_new}). It is trivial to check that the optimality
conditions for (\ref{eq:market_new}) correspond to (\ref{eq:mk}).
Formulation (\ref{eq:market_new}) is clearly stable, because the
market clearing constraint is satisfied precisely. The equality constraint
on the dual variables makes sure that the optimal solution remains
the same if the market clearing constraint is removed after the calculation
of the optimal solution. Formulation (\ref{eq:market_new}) is used
as a definition of the hypothetical market agent in the rest of this
work.

We can see that, by affecting the expected electricity price, the
hypothetical agent changes the electricity price process. It is not
immediately clear how to construct such a stochastic process or that
such a stochastic process exists at all. We refer the reader to \cite{troha2014theexistence},
where a constructive proof of the existence is given. The proof is
based on the Doob decomposition theorem, where we allow the hypothetical
market agent to control an integrable predictable term of the process,
while keeping the martingale term of the process intact.

For the further argumentation we define $v_{P}=\left|\right|_{p\in P}v_{p}$
and $v=\left[v_{P}^{\top},V_{C}^{\top}\right]^{\top}$.

\subsection{Nash equilibrium\label{sec:Analysis}}

Binarity constraints (\ref{eq:w2_2_b}), (\ref{eq:w4_2_b}), (\ref{eq:w5_2_b}),
(\ref{eq:b_7}) and (\ref{eq:b_8}) of each producer significantly
complicate the analysis of Problem (\ref{eq:opt_prod}) and thus,
we focus on the continuous relaxation (\ref{eq:opt_prod_rel}) instead.
We then show through various numerical results in Section \ref{sec:N1}
and Section \ref{sub:N2}, that binarity constraints (\ref{eq:w2_2_b}),
(\ref{eq:w4_2_b}), (\ref{eq:w5_2_b}), (\ref{eq:b_7}) and (\ref{eq:b_8})
do not have a significant impact on the equilibrium electricity price. 

Using the continuous relaxation (\ref{eq:opt_prod_rel}), we are interested
in finding a Nash equilibrium defined as

\begin{definition}\label{Nash-Equilibrium-(NE)}Nash Equilibrium
(NE)

Decisions $v^{*}$ and $\mathbb{E}^{\mathbb{P}}\left[\Pi\right]^{*}$
constitute a Nash equilibrium if
\begin{enumerate}
\item For every producer $p\in P$, $v_{p}^{*}$ is a strategy such that
\begin{equation}
\Psi_{p}\left(v_{p},\mathbb{E}^{\mathbb{P}}\left[\Pi\right]^{*}\right)\leq\Psi_{p}\left(v_{p}^{*},\mathbb{E}^{\mathbb{P}}\left[\Pi\right]^{*}\right)
\end{equation}
for all $v_{p}\in\widetilde{S}_{p}$;
\item For every consumer $c\in C$, $V_{c}^{*}$ is a strategy such that
\begin{equation}
\Psi_{c}\left(V_{c},\mathbb{E}^{\mathbb{P}}\left[\Pi\right]^{*}\right)\leq\Psi_{c}\left(V_{c}^{*},\mathbb{E}^{\mathbb{P}}\left[\Pi\right]^{*}\right)
\end{equation}
for all $V_{c}\in S_{c}$;
\item Price vector $\mathbb{E}^{\mathbb{P}}\left[\Pi\right]^{*}$ maximizes
the objective function of the hypothetical market agent, i.e. 
\begin{equation}
\Psi_{M}\left(\mathbb{E}^{\mathbb{P}}\left[\Pi\right],v^{*}\right)\leq\Psi_{M}\left(\mathbb{E}^{\mathbb{P}}\left[\Pi\right]^{*},v^{*}\right)\label{eq:market}
\end{equation}
for all $\mathbb{E}^{\mathbb{P}}\left[\Pi\right]\in S_{M}$.
\end{enumerate}
\end{definition}

From Definition (\ref{Nash-Equilibrium-(NE)}), it is not clear whether
a NE for our problem exists and whether it is unique. This problem
was thoroughly investigated in \cite{troha2014theexistence}. Roughly
speaking, it was shown that if the demand of the end users can be
covered by the available system of power plants, then a NE exists.
Moreover, if the power plants are similar enough (if there are no
big gaps in the efficiency of the power plants), then one can show
that the NE is also unique. On the other hand, if power plants are
similar enough, then the expected equilibrium price of each electricity
contract might be an interval instead of a single point. 

In this paper we focus on the numerical calculation of the NE under
the assumption of the existence of solution. For this paper, we assume
the following, a slightly stricter, condition.

\begin{assumption}\label{ass: as1} For all $p\in P$, the exists
vector $v_{p}$ such that $A_{p}v_{p}=a_{p}$ a.s. and $B_{p}v_{p}<b_{p}$
a.s., for all $c\in C$, there exists vector $V_{c}$ such that $A_{c}V_{c}=a_{c}$
a.s. and $B_{c}V_{c}<b_{c}$ a.s., and the vectors $V_{p}$ and $V_{c}$
can be chosen so that (\ref{eq:KKT_hyp_mark}) is satisfied.\end{assumption}

\subsection{Quadratic programming formulation\label{sec:Algorithm}}

The traditional approach to solving equilibrium optimization problems
is through shadow prices (see \cite{demaeredaertrycke2012liquidity}
for example). However, this approach is only valid when no inequality
constraints are present. Shadow prices depend on the set of active
constraints and thus one can only use this approach when the active
set is known. In inequality constrainted optimization, the active
set is usually not know in advance and thus a different approach is
needed. The proposed formulation below can be seen as an extension
of the shadow price concept to inequality constrained optimization
problems. 

A naive approach for solving inequality constrained equilibrium optimization
problem would be to choose an expected price vector $\mathbb{E}^{\mathbb{P}}\left[\Pi\right]$
and then calculate optimal solutions for each producer $p\in P$ and
each consumer $c\in C$ by solving (\ref{eq:opt_prod_rel}) and (\ref{eq:opt_con}),
respectively. If at such price $\left\Vert \sum_{c\in C}V_{c}+\sum_{p\in P}V_{p}\right\Vert $
is close to zero, then the solution is found and $\mathbb{E}^{\mathbb{P}}\left[\Pi\right]$
is an equilibrium expected price vector. Otherwise, we have to adjust
the expected price vector and repeat the procedure. We can see that
such an algorithm is costly, because it requires to solve a large
optimization problem (i.e. to calculate the optimal solutions of each
producer and each consumer) multiple times. In the section below,
we show that we can do much better than the naive approach. Using
the reformulation we propose, the large optimization problem must
be solved only once. 

Necessary and sufficient conditions for all $v_{k}$, $k\in P\cup C$
and $\mathbb{E}^{\mathbb{P}}\left[\Pi\right]$ to constitute a NE
are the following, due to the fact that Assumption \ref{ass: as1}
implies the Slater condition, 
\begin{equation}
\begin{array}{rcl}
-\mathbb{E}^{\mathbb{P}}\left[\pi_{k}\right]^{\top}-\lambda_{k}Q_{k}v_{k}-B_{k}^{\top}\eta_{k}-A_{k}^{\top}\mu_{k} & = & 0\\
\\
\eta_{k}^{\top}\left(B_{k}v_{k}-b_{k}\right) & = & 0\\
\\
B_{k}v_{k}-b_{k} & \leq & 0\\
\\
A_{k}v_{k}-a_{k} & = & 0\\
\\
\eta_{k} & \geq & 0\\
\\
\sum_{c\in C}V_{c}+\sum_{p\in P}V_{p} & = & 0.
\end{array}\label{eq:KKT_all-2}
\end{equation}
The last equation corresponds to the KKT conditions of the hypothetical
market agent. 

We can now interpret (\ref{eq:KKT_all-2}) as the KKT conditions of
one large optimization problem that includes the new definition (\ref{eq:market_new})
of the hypothetical market agent. To see this, we join all decision
variables into one vector $x:=\left[v^{\top},\mathbb{E}^{\mathbb{P}}\left[\Pi\right]^{\top}\right]^{\top}$
and rewrite
\begin{itemize}
\item the equality constraints as $Ax=a$ with $a:=\left[a_{p_{1}}^{\top},...,a_{p_{P}}^{\top},a_{c_{1}}^{\top},...,a_{c_{C}}^{\top},\underbrace{0,...,0}_{N}\right]^{\top}$where
the number of ending zeros is equal to $N$, and 
\[
A:=\left[\begin{array}{ccccccc}
A_{p_{1}} & 0 &  &  &  &  & 0\\
0 & \ddots & 0 &  &  &  & \vdots\\
 & 0 & A_{p_{P}} & 0 &  &  & 0\\
 &  & 0 & A_{c_{1}} & 0 &  & 0\\
 &  &  & 0 & \ddots & 0 & \vdots\\
 &  &  &  & 0 & A_{c_{C}} & 0\\
M_{p_{1}} & \cdots & M_{p_{P}} & I & \cdots & I & 0
\end{array}\right],
\]
where $M_{p}\in\mathbb{R}^{N}\times\mathbb{R}^{\dim v_{p}}$ is a
matrix defined as 
\[
M_{p}=\left[\text{diag}\left(\underset{N}{\underbrace{1,...,1}}\right)\left|\begin{array}{ccc}
0 & \cdots & 0\\
\vdots & \ddots & \vdots\\
0 & \cdots & 0
\end{array}\right.\right],
\]

\item the inequality constraints as $Bx\leq b$ with $b:=\left[b_{p_{1}}^{\top},...,b_{p_{P}}^{\top},b_{c_{1}}^{\top},...,b_{c_{C}}^{\top}\right]^{\top}$,
and
\[
B:=\left[\begin{array}{ccccccc}
B_{p_{1}} & 0 &  &  &  &  & 0\\
0 & \ddots & 0 &  &  &  & \vdots\\
 & 0 & B_{p_{P}} & 0 &  &  & 0\\
 &  & 0 & B_{c_{1}} & 0 &  & 0\\
 &  &  & 0 & \ddots & 0 & \vdots\\
 &  &  &  & 0 & B_{c_{C}} & 0
\end{array}\right],
\]

\item the objective function as $-\pi^{\top}x-\frac{1}{2}x^{\top}Qx$ with
$\pi:=\left[\mathbb{E}^{\mathbb{P}}\left[\pi_{0,p_{1}}\right]^{\top},...,\mathbb{E}^{\mathbb{P}}\left[\pi_{0,p_{P}}\right]^{\top},\underset{\left(\left|C\right|+1\right)N}{\underbrace{0,...,0}}\right]^{\top}$
where $\pi_{0,p}$ is $\pi_{p}$ with elements of $\Pi$ set to zero,
and
\begin{equation}
Q:=\left[\begin{array}{ccccccc}
\lambda_{p_{1}}Q_{p_{1}} & 0 &  &  &  &  & M_{p_{1}}^{\top}\\
0 & \ddots & 0 &  &  &  & \vdots\\
 & 0 & \lambda_{p_{P}}Q_{p_{P}} & 0 &  &  & M_{p_{P}}^{\top}\\
 &  & 0 & \lambda_{c_{1}}Q_{c_{1}} & 0 &  & I\\
 &  &  & 0 & \ddots & 0 & \vdots\\
 &  &  &  & 0 & \lambda_{c_{C}}Q_{c_{C}} & I\\
M_{p_{1}} & \cdots & M_{p_{P}} & I & \cdots & I & 0
\end{array}\right],\label{eq:Q-1}
\end{equation}

\item the dual variables as $\eta:=\left[\eta_{p_{1}}^{\top},...,\eta_{p_{P}}^{\top},\eta_{c_{1}}^{\top},...,\eta_{c_{C}}^{\top}\right]$
and $\mu:=\left[\mu_{p_{1}}^{\top},...,\mu_{p_{P}}^{\top},\mu_{c_{1}}^{\top},...,\mu_{c_{C}}^{\top},\mu_{M}^{\top}\right]$. 
\end{itemize}
In this setting we can reformulate the KKT conditions (\ref{eq:KKT_all-2})
as follows,
\begin{equation}
\begin{array}{rcl}
-\pi-Qx-B^{\top}\eta-A^{\top}\mu & = & 0\\
\\
\eta^{\top}\left(Bx-b\right) & = & 0\\
\\
Bx-b & \leq & 0\\
\\
Ax-a & = & 0\\
\\
\eta & \geq & 0\\
\\
\mu_{M} & = & 0.
\end{array}\label{eq:KKT_all-3}
\end{equation}

Since the additional constraints $\mu_{M}=0$ on the dual variables
of Problem \ref{eq:KKT_all-3} cannot be handled by most of the available
quadratic programming solvers, we have to reformulate the problem
in a dual form. We start by formulating the optimization problem out
of the KKT conditions (\ref{eq:KKT_all-3}) as 
\begin{equation}
\begin{array}{cl}
\underset{x}{\text{max}} & -\pi^{\top}x-\frac{1}{2}x^{\top}Qx\\
\\
\text{s.t.} & Ax=a\\
\\
 & Bx\leq b\\
\\
 & \mu_{M}=0
\end{array}\label{eq:op_all-2}
\end{equation}
and by defining the Lagrangian as

\[
\mathcal{L}\left(x,\mu,\eta\right)=\left\{ \begin{array}{rl}
-\frac{1}{2}x^{\top}Qx-\pi^{\top}x-\left(Ax-a\right)^{\top}\mu-\left(Bx-b\right)^{\top}\eta; & \text{if }\eta\geq0\\
\\
-\infty; & \text{otherwise}.
\end{array}\right.
\]
One can show that, $Q\succeq0$ for all vectors that satisfy the market
clearing constraint (\ref{eq:mk}) (for the proof see \cite{troha2014calculation}).
$\mathcal{L}\left(x,\mu,\eta\right)$ is therefore a smooth and convex
function. The unconstrained minimizer can be determined by solving
$\mathcal{D}_{x}\mathcal{L}\left(x,\mu,\eta\right)=0$. Calculating
\[
\mathcal{D}_{x}\mathcal{L}\left(x,\mu,\eta\right)=-Qx-\pi-A^{\top}\mu-B^{\top}\eta
\]
and inserting $\pi$ back to the Lagrangian, an equivalent formulation
is obtained as follows,
\[
\mathcal{L}\left(x,\mu,\eta\right)=\left\{ \begin{array}{rl}
\frac{1}{2}x^{\top}Qx+a^{\top}\mu+b^{\top}\eta & \text{if }\eta\geq0\text{ \text{and }}-Qx-\pi-A^{\top}\mu-B^{\top}\eta=0,\\
\\
-\infty & \text{otherwise}
\end{array}\right.
\]
Relating the latter to a maximization optimization problem, the following
formulation is obtained 
\begin{equation}
\begin{array}{cl}
\underset{x,\mu,\eta}{\text{max}} & -\frac{1}{2}x^{\top}Qx-\mu^{\top}a-\eta^{\top}b\\
\\
\text{s.t.} & Qx+A^{\top}\mu+B^{\top}\lambda+\pi=0\\
\\
 & \eta\geq0\\
\\
 & \mu_{M}=0.
\end{array}\label{eq:op_all-1-1}
\end{equation}
Problem (\ref{eq:op_all-1-1}) is equivalent to Problem (\ref{eq:op_all-2}),
but it can be solved using any quadratic programming algorithm. 

Based on our discussion in Section \ref{sec:Analysis}, we can see
that (\ref{eq:op_all-1-1}) was obtained by considering Problem (\ref{eq:opt_prod_rel}),
which is a continuous relaxation of Problem (\ref{eq:opt_prod}).
To estimate the error caused by the continuous relaxation, we use
the following procedure:
\begin{enumerate}
\item We calculate the equilibrium electricity price $\mathbb{E}^{\mathbb{P}}\left[\Pi\right]^{*}$
by solving problem (\ref{eq:op_all-1-1}).
\item Using the equilibrium electricity price $\mathbb{E}^{\mathbb{P}}\left[\Pi\right]^{*}$
from the previous step, we calculate optimal trading vectors $V_{p}^{*}$,
$p\in P$ for all producers and optimal trading vectors $V_{c}^{*}$,
$c\in C$ for all consumers by solving (\ref{eq:opt_prod}) and (\ref{eq:opt_con}),
respectively.
\item We calculate the error as
\begin{equation}
\text{MIQP}:=\sum_{c\in C}V_{c}^{*}+\sum_{p\in P}V_{p}^{*}.\label{eq:MIQP}
\end{equation}
 
\end{enumerate}
In order to verify the procedure above, we apply the following very
similar procedure:
\begin{enumerate}
\item We calculate the equilibrium electricity price $\mathbb{E}^{\mathbb{P}}\left[\Pi\right]^{*}$
by solving problem (\ref{eq:op_all-1-1}).
\item Using the equilibrium electricity price $\mathbb{E}^{\mathbb{P}}\left[\Pi\right]^{*}$
from the previous step, we calculate optimal trading vectors $V_{p}^{*}$,
$p\in P$ for all producers and optimal trading vectors $V_{c}^{*}$,
$c\in C$ for all consumers by solving (\ref{eq:opt_prod_rel}) and
(\ref{eq:opt_con}), respectively.
\item We calculate the error as
\begin{equation}
\text{QP}:=\sum_{c\in C}V_{c}^{*}+\sum_{p\in P}V_{p}^{*}.\label{eq:QP}
\end{equation}

\end{enumerate}
In Section \ref{sec:N1} and Section \ref{sub:N2}, we present the
MIQP and QP when modeling the entire UK power grid.

\section{Numerical results\label{sec:N1}}

In this section we discuss the numerical results and apply our model
from Section \ref{sec:Problem-description} to model the realistic
UK power grid.

\subsection{Estimation of parameters\label{sub:Estimation-of-parameters}}

In this section we investigate how to estimate various parameters
of power plants that enter our model described in Section \ref{sec:Problem-description}. 

In the UK all power plants are required to submit their available
capacity as well as ramp-up and ramp-down constraints to the grid
operator on a half hourly basis. This data is publicly available at
the Elexon website%
\footnote{http://www.bmreports.com/%
}. A more challenging problem is to estimate of the efficiency $c^{p,l,r}$,
startup costs $s^{p,l,r}$ and the carbon emission intensity factor
$g^{p,l,r}$ for each power plant $r\in R^{p,l}$. For the purpose
of the calibration, we assume that all producers are risk neutral
and set $\lambda_{p}=0$ for all $p\in P$. Furthermore, we neglect
the ramp-up and ramp-down constraints (\ref{eq:ramp}), (\ref{eq:ramp_up_eq1}),
(\ref{eq:ramp_up_eq2}), (\ref{eq:ramp_down_eq1}), and (\ref{eq:ramp_down_eq2}).
Since each power plant is treated separately, we avoid writing subscripts/superscripts
$p,l,r$.

Before we explore the details of the calibration process, let us establish
a few relationships that will prove useful later in this section.
We can see that a power plant will produce at time $T_{j}$ if the
income from selling electricity at the spot price is greater than
the costs of purchasing the required fuel and emission certificates
at the current spot price (remember that a power plant has to cover
the startup costs too). Thus, for a power plant that runs on fuel
$l\in L$ and produces electricity at time $T_{j}$,
\begin{equation}
\Pi\left(T_{j},T_{j}\right)-cG_{l}\left(T_{j},T_{j}\right)-gG_{em}\left(T_{j},T_{j}\right)>0\label{eq:e2}
\end{equation}
must hold for production to take place.

It is immediately clear why (\ref{eq:e2}) must hold when only spot
contracts are available. Let us investigate why (\ref{eq:e2}) holds
also if forward and future electricity contracts are available on
the market. At any trading time $t_{i}$, $i\in I_{j}$, a rational
producer could enter into a short electricity forward contract and
simultaneously into a long fuel and emission forward contract if
\begin{equation}
\Pi\left(t_{i},T_{j}\right)-cG_{l}\left(t_{i},T_{j}\right)-gG_{em}\left(t_{i},T_{j}\right)>0.\label{eq:prodIneq-1}
\end{equation}
At delivery time $T_{j}$, this producer has two options:
\begin{itemize}
\item To acquire the delivery of the fuel and emission certificates bought
at trading time $t_{i}$ and produce electricity. In this case, she
observes the following profit 
\begin{equation}
\widehat{P_{1}}\left(T_{j}\right)=\Pi\left(t_{i},T_{j}\right)-cG_{l}\left(t_{i},T_{j}\right)-gG_{em}\left(t_{i},T_{j}\right).
\end{equation}

\item To produce no electricity and instead close the forward electricity,
fuel, and emission contracts. In this case, she observes the following
profit 
\begin{equation}
\begin{array}{rcl}
\widehat{P_{2}}\left(T_{j}\right) & = & \left[\Pi\left(t_{i},T_{j}\right)-\Pi\left(T_{j},T_{j}\right)\right]-c\left[G_{l}\left(t_{i},T_{j}\right)-G_{l}\left(T_{j},T_{j}\right)\right]\\
\\
 &  & -g\left[G_{em}\left(t_{i},T_{j}\right)-G_{em}\left(T_{j},T_{j}\right)\right].
\end{array}
\end{equation}

\end{itemize}
Power plant $r\in R^{p,l}$ will run at $T_{j}$ if and only if
\begin{equation}
\widehat{P_{1}}\left(T_{j}\right)>\widehat{P_{2}}\left(T_{j}\right).\label{eq:e1}
\end{equation}
With some reordering of the terms, it is easy to see that inequality
(\ref{eq:e1}) is equivalent to inequality (\ref{eq:e2}).

Using the reasoning above, we can conclude that, for the purpose of
determining the stack, it is enough to focus only on spot electricity,
fuel and emission contracts. By taking into account startup costs
and equations described in Section \ref{sub:Producer}, the profit
maximization problem of each power plant can be written as
\begin{equation}
\underset{W^{\left(2\right)},W^{\left(4\right)},W^{\left(6\right)}}{\max}\sum_{j\in J}\widehat{W}\left(T_{j}\right)\overline{P}\left(T_{j}\right)-W^{\left(4\right)}\left(T_{j}\right)s\label{eq:mg}
\end{equation}
subject to
\begin{align}
W^{\left(4\right)}\left(T_{j}\right)\geq W^{\left(2\right)}\left(T_{j}\right)-W^{\left(2\right)}\left(T_{j-1}\right), & \forall j\in J\backslash\left\{ 1\right\} \label{eq:C1}\\
W^{\left(6\right)}\left(T_{j}\right)\leq W^{\left(2\right)}\left(T_{j}\right), & \forall j\in J\label{eq:C2}\\
W^{\left(k\right)}\left(T_{j}\right)\in\left[0,1\right], & \forall j\in J,\; k\in\left\{ 2,4,6\right\} \label{eq:C3}\\
W^{\left(2\right)}\left(T_{j}\right)\in\mathbb{Z}, & \forall j\in J,\label{eq:C4}
\end{align}
where 
\begin{equation}
\widehat{W}\left(T_{j}\right)=W^{\left(2\right)}\left(T_{j}\right)\overline{W}_{min}\left(T_{j}\right)+W^{\left(6\right)}\left(T_{j}\right)\left(\overline{W}_{max}\left(T_{j}\right)-\overline{W}_{min}\left(T_{j}\right)\right),\;\forall j\in J\label{eq:allW}
\end{equation}
and 
\begin{equation}
\bar{P}\left(T_{j}\right)=\Pi\left(T_{j},T_{j}\right)-cG\left(T_{j},T_{j}\right)-gG_{em}\left(T_{j},T_{j}\right),\;\forall j\in J.\label{eq:costs}
\end{equation}
Note that we do not have to impose the integrality constraints for
variable $W^{\left(4\right)}\left(T_{j}\right)$ $j\in J$, because
they are implied by (\ref{eq:C4}) and (\ref{eq:C1}).
To account for the neglected risk premium, trading costs, maintenance
costs etc. we introduce an additional constant $m>0$ and include
it in (\ref{eq:costs}) as
\begin{equation}
\bar{P}\left(T_{j}\right)=\Pi\left(T_{j},T_{j}\right)-cG\left(T_{j},T_{j}\right)-gG_{em}\left(T_{j},T_{j}\right)-m,\;\forall j\in J.\label{eq:costs-1}
\end{equation}

We are interested to know how the optimal solution of Problem (\ref{eq:mg})
depends on parameters $c$, $g$, $m$, and $s$. Let $\widehat{W}^{*}\left(T_{j};c,g,m,s\right)$
denote the optimal production of Problem (\ref{eq:mg}). Our task
is to find $c$, $g$, $m$, and $s$ that satisfy
\begin{equation}
\underset{c,g,m,s}{\min}\sum_{j\in J}\left(\widehat{W}^{*}\left(T_{j};c,g,m,s\right)-\tilde{W}\left(T_{j}\right)\right)^{2}\label{eq:outer}
\end{equation}
where $\tilde{W}\left(T_{j}\right)$ denotes observed historical production
of a power plant. The optimization problem is a bi-level optimization
problem where (\ref{eq:outer}) corresponds to the outer optimization
problem and (\ref{eq:mg}) corresponds to the inner optimization problem.
Traditionally, such problems have been very difficult to solve, because
they are highly non-convex and the process of finding the optimal
solution of the outer optimization problem requires many expensive
evaluations of the inner integer programming optimization problem.
However, we can show that in our case a difficult integer programming
problem can be replaced by a tractable linear programming problem
without affecting the optimal solution.

We can use the following proposition to see that optimal solution
of Problem (\ref{eq:mg}) can be calculated by a linear programming
relaxation.

\begin{proposition}\label{prop:TU}The matrix of constraints (\ref{eq:C1}),
(\ref{eq:C2}), and (\ref{eq:C3}) for Problem (\ref{eq:mg}) is totally
unimodular.\end{proposition}

\begin{proof}Let us write the matrix of inequality constraints (\ref{eq:C1})
and (\ref{eq:C2}) as 
\begin{equation}
\left[A_{1}A_{2}A_{3}\right]\left[\begin{array}{c}
W^{\left(2\right)}\\
W^{\left(6\right)}\\
W^{\left(4\right)}
\end{array}\right]\leq0
\end{equation}
for some block matrices $A_{1}$, $A_{2}$, and $A_{3}$. We will
first show that matrix $\left[A_{1}A_{2}\right]$ is totally unimodular.
Note that all entries are $\left\{ -1,0,1\right\} $. Moreover, each
row contains exactly two non-zero entries. One of the entries is $1$
and the other is $-1$. These are sufficient conditions for matrix
$\left[A_{1}A_{2}\right]$ to be totally unimodular. It is trivial
to see that $A_{3}=P\left[\begin{array}{cc}
I & 0\\
0 & 0
\end{array}\right]Q$ for some permutation matrices $P$ and $Q$ of the appropriate size.
This implies that matrix $\left[A_{1}A_{2}A_{3}\right]$ is totally
unimodular. The bound constraints (\ref{eq:C3}) can be included by
using a similar argument.\end{proof}

By the virtue of Proposition \ref{prop:TU}, we can relax the binarity
constraints and reformulate Problem (\ref{eq:mg}) as an linear programming
problem as
\begin{equation}
\begin{array}{rl}
\underset{W^{\left(2\right)},W^{\left(4\right)},W^{\left(6\right)}}{\max} & \sum_{j\in J}\widehat{W}\left(T_{j}\right)\overline{P}\left(T_{j}\right)-W^{\left(4\right)}\left(T_{j}\right)s\\
\\
\text{s.t.} & W^{\left(4\right)}\left(T_{j}\right)\geq W^{\left(2\right)}\left(T_{j}\right)-W^{\left(2\right)}\left(T_{j-1}\right),\;\forall j\in J\backslash\left\{ 1\right\} \\
\\
 & W^{\left(6\right)}\left(T_{j}\right)\leq W^{\left(2\right)}\left(T_{j}\right),\;\forall j\in J\\
\\
 & W^{\left(k\right)}\left(T_{j}\right)\in\left[0,1\right],\; k\in\left\{ 2,4,6\right\} ,\:\forall j\in J.
\end{array}\label{eq:mg-1}
\end{equation}
A combination of a particle swarm algorithm \cite{vaz2007aparticle}
and Gurobi \cite{gurobioptimization2014gurobioptimizer} was used
to solve the bi-level optimization problem (\ref{eq:mg-1}) in practice.
Particle swarm was applied to the outer and Gurobi to the inner optimization
problem.

For each power plant we used over 5000 training samples obtained from
the period between 1/1/2012 and 1/1/2013.

\subsection{UK power grid}

In this section we apply our model to the entire system of the UK
power plants. We focus on the coal, gas, and oil power plants, because
these power plants adapt their production to cover the changes in
demand and are thus responsible for setting the price. Nuclear power
plants do not have to be modeled explicitly because their ramp-up
and ramp-down constraints are so tight that their production is almost
constant over time. They usually deviate from the maximum production
only for maintenance reasons. Renewable sources and interconnectors
are not modeled explicitly, because they require a different treatment
not covered in this paper. In this section, we define demand $D\left(T_{j}\right)$
for all $j\in J$ as
\begin{equation}
D\left(T_{j}\right):=D_{act}\left(T_{j}\right)-P_{renw}\left(T_{j}\right)-P_{inter}\left(T_{j}\right)
\end{equation}
where $D_{act}\left(T_{j}\right)$ denotes the actual demand in the
UK power system, $P_{renw}\left(T_{j}\right)$ denotes the production
from all renewable sources including wind, solar, biomass, hydro and
pumped storage, and $P_{inter}\left(T_{j}\right)$ denotes the inflow
of power into the UK power system through interconnectors. To make
this model useful in practice one has to model each of these terms,
but this exceeds the scope of this paper.

Our goal is to calculate the electricity spot price with the information
available on 11/2/2013. We are interested in a delivery period from
4/4/2013 00:00:00 to 8/4/2013 00:00:00. We assume that there are two
types of power contract available. The first is a month ahead contract
traded on 15/3/2013 17:00:00 and covers the delivery over all four
days. The second type is a spot contract that requires an immediate
delivery and is traded for each half hour separately. We use future
prices of coal, gas, and oil as available on 11/2/2013. Since the
historical demand forecast is not available, we used the realized
demand instead, which is a standard practice in the literature. To
use this model in practice, one could use a demand forecast available
at the Elexon webpage%
\footnote{http://www.bmreports.com/%
} or develop a new approach. Since we do not have the information about
the ownership of the power plants, we assumed that there is only one
producer who owns all power plants connected to the UK grid and only
one consumer that is responsible for satisfying the demand of the
end users. In reality, market participants have more information about
the ownership that can be incorporated into the model. We set $\lambda_{k}=10^{-7}$
for all $k\in P\cup C$. The impact of the risk aversion of producers
and consumers is thoroughly investigated in \cite{troha2014calculation}.
As described in the previous section, we estimated parameters $c$,
$g$, $m$, and $s$ for each power plant from 5000 training samples
obtained in the period between 1/1/2012 and 1/1/2013.

To motivate the inclusion of startup costs we first investigate a
simplified version of our model described in Section \ref{sec:Problem-description}
and neglect the startup costs. Figure \ref{fig:f8} shows the output
of our model, when all startup costs are set to zero. The figure on
the left hand side depicts the calculated energy mix between coal
and gas power plants, while the figure on the right hand side depicts
the actually observed energy mix. Both figures contain also the spot
price calculated by our model and the actually observed spot price.
The difference between calculated and observed production for each
fuel is depicted in Figure \ref{fig:f5_1}. We can see that our model
predicts the energy mix very closely. Moreover, the daily pattern
of the electricity price predicted by our model is similar to the
actually observed one. The model correctly predicted that the electricity
price is higher during the peak hours than during the off peak hours.
Furthermore, the calculated electricity price has two daily peaks
that occur at almost the same time as in the historically observed
price. 

The graphs also reveal a few problems of our model. Firstly, we can
see that our model underestimates spot prices during peak hours and
overestimates them during the off-peak hours. A similar results was
also found in \cite{harvey2002marketpower}. Secondly, the two spikes
in the observed price are not captured in our model. This motivated
us to extend our model and incorporate the startup costs of the power
plants. For the purpose of calibration, we applied the approach described
in Section \ref{sub:Estimation-of-parameters}. 

\begin{figure}
\begin{centering}
\includegraphics[bb=35bp 180bp 545bp 600bp,clip,scale=0.43]{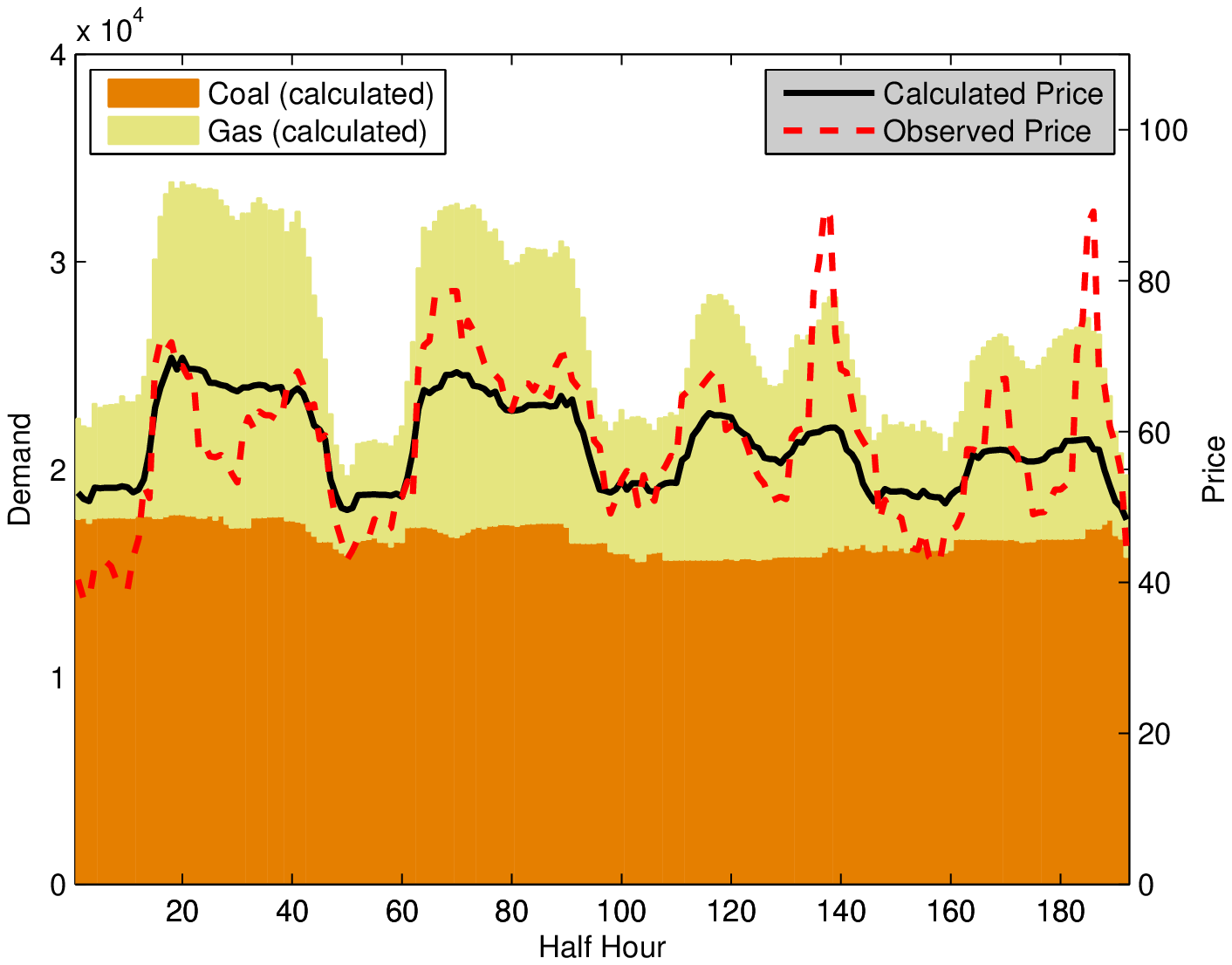}\includegraphics[bb=35bp 180bp 545bp 600bp,clip,scale=0.43]{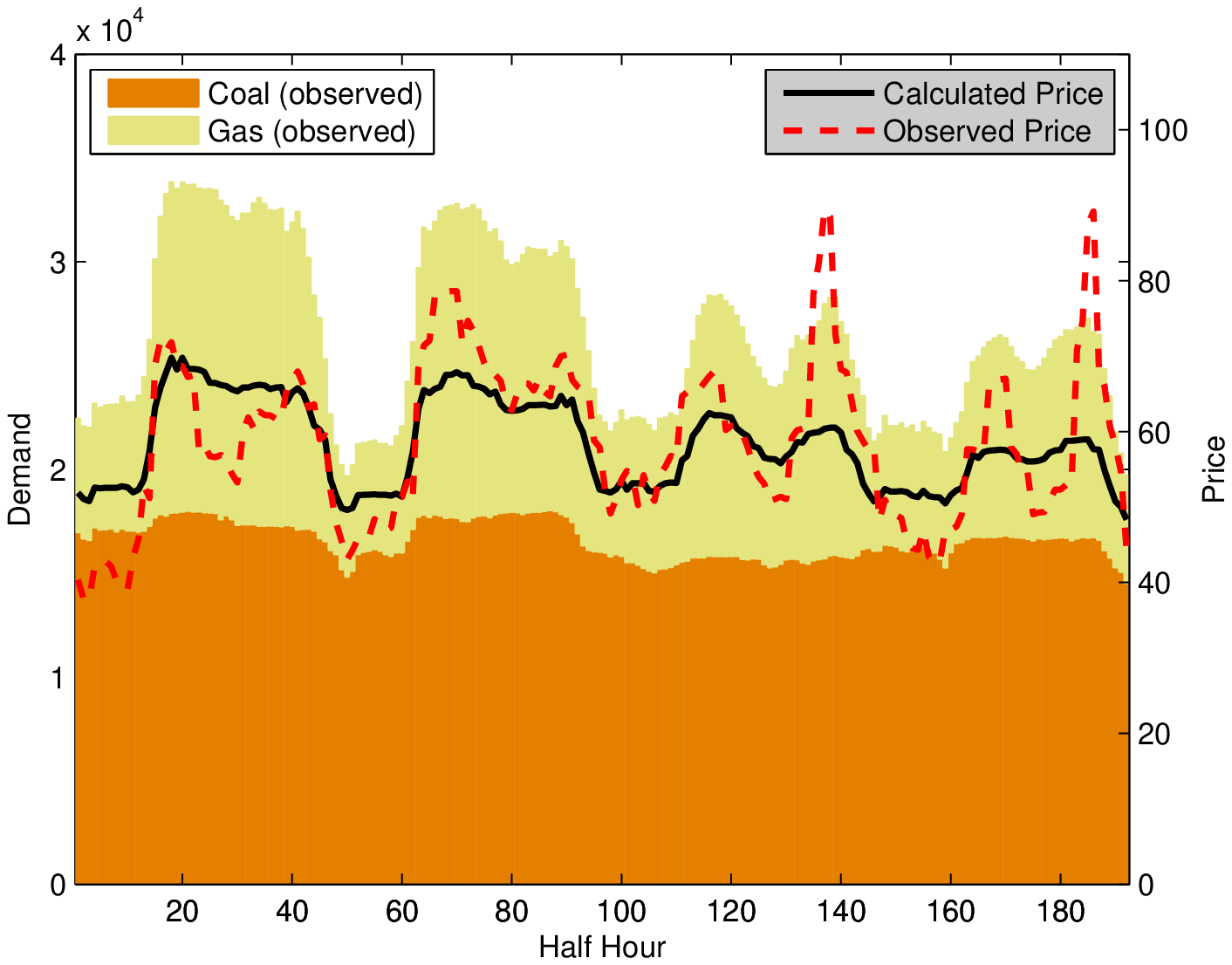}
\par\end{centering}

\caption{\label{fig:f8}Comparison of the calculated and historical electricity
price and energy mix when startup costs are excluded (i.e. set to
zero). }
\end{figure}

\begin{figure}
\begin{centering}
\includegraphics[scale=0.43]{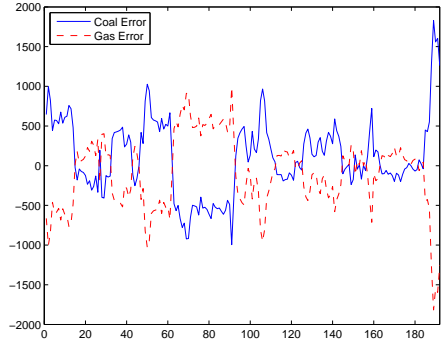}
\par\end{centering}

\caption{\label{fig:f5_1}The difference between calculated and observed gas
and coal production.}

\end{figure}

Calculated equilibrium prices and the energy mix with startup costs
included are depicted in Figure \ref{fig:f8-1}. By comparing Figure
\ref{fig:f8} and Figure \ref{fig:f8-1}, we can see that the calculated
equilibrium price captures the daily variations of the actually observed
price much more closely. It correctly predicts some of the spikes,
but also forecasts many false positives. Figure \ref{fig:f5_1-1}
shows that the inclusion of startup costs slightly improved the error
in the energy mix calculation.

\begin{figure}
\begin{centering}
\includegraphics[bb=35bp 180bp 545bp 600bp,clip,scale=0.43]{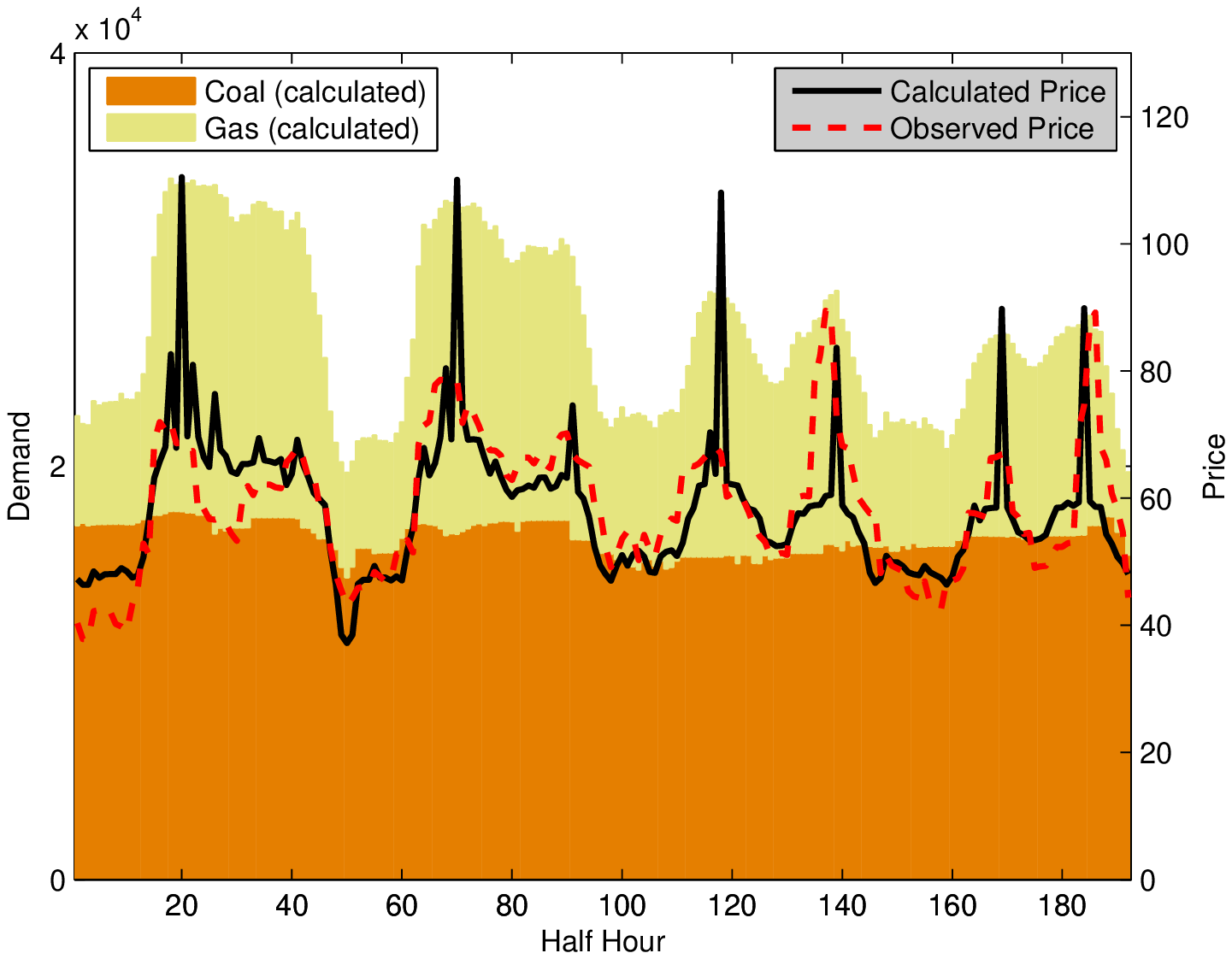}\includegraphics[bb=35bp 180bp 545bp 600bp,clip,scale=0.43]{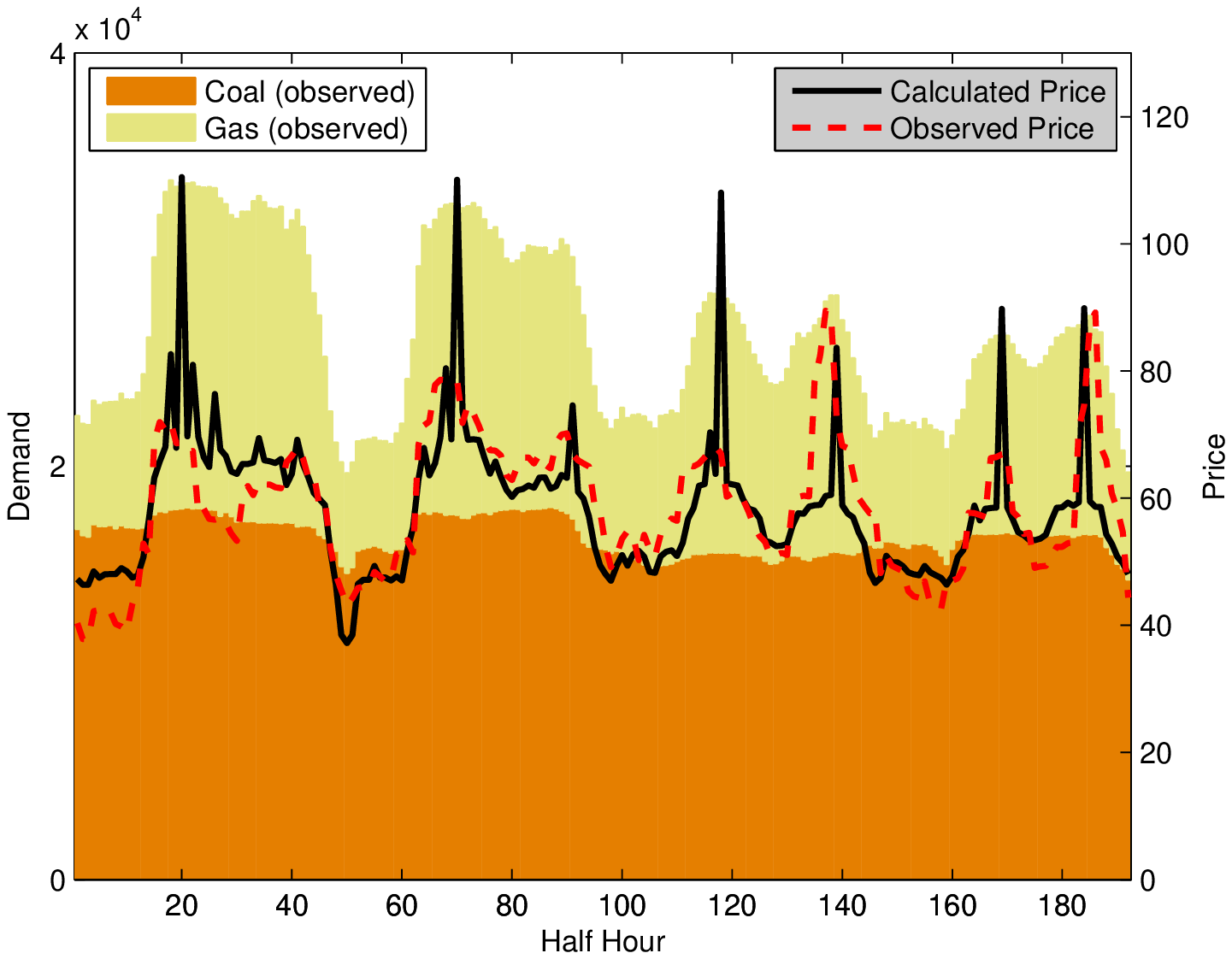}
\par\end{centering}

\caption{\label{fig:f8-1}Comparison of the calculated and historical electricity
price and energy mix when startup costs are included.}
\end{figure}
\begin{figure}
\begin{centering}
\includegraphics[scale=0.43]{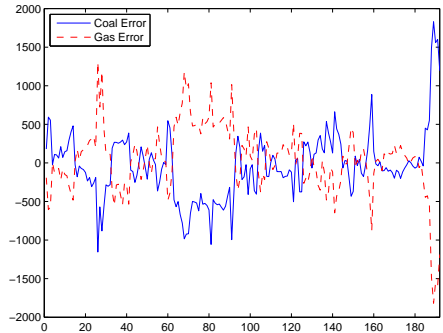}
\par\end{centering}

\caption{\label{fig:f5_1-1}The difference between calculated and observed
gas and coal production after including startup costs.}
\end{figure}

It is interesting to explore the conditions of the electricity grid
at times when the spikes in the electricity price occur. A very descriptive
parameter is standing reserve $SR\left(T_{j}\right)$, $j\in J$,
defined as
\begin{equation}
SR\left(T_{j}\right)=\sum_{p\in P}\sum_{l\in L}\sum_{r\in R^{p,l}}\left[W_{p,l,r}^{\left(2\right)}\left(T_{j}\right)-W_{p,l,r}^{\left(6\right)}\left(T_{j}\right)\right]\left[\overline{W}_{max}^{p,l,r}\left(T_{j}\right)-\overline{W}_{min}^{p,l,r}\left(T_{j}\right)\right],
\end{equation}
which quantifies by how much the power plants that are currently running
can increase their production before a new power plant must be turned
on. Since most of the power plants have severe constraints on startup
times, low standing reserve usually implies low stability of the electricity
grid.

Figure \ref{fig:f8-1-1} depicts the calculated standing reserve over
the relevant time period. We can see that all price spikes occur when
standing reserve is close to zero. In such situations, a new power
plant must be turned on (and off quickly afterwards) to cover the
temporary extra demand. Thus, the startup costs are spread over a
very short period of time, and a high electricity price is required
for such an action to be profitable. However, in reality, the times
of a low standing reserve are very rare. The grid operator is responsible
for providing a reliable electricity delivery and preventing times
with a low standing reserve. This is achieved by incentivizing some
of the power plants to start production even when it is not profitable
for them. The costs of such actions are distributed among all market
participants. How to include the grid operator in our model is discussed
in the next section.

\begin{figure}
\begin{centering}
\includegraphics[bb=35bp 180bp 545bp 600bp,clip,scale=0.43]{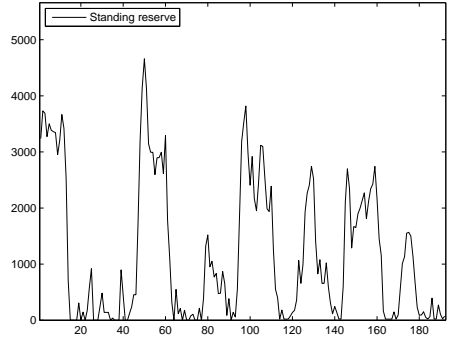}
\par\end{centering}

\caption{\label{fig:f8-1-1}Standing reserve for the relevant time period.}
\end{figure}

In the remaining part of this section, we evaluate the error caused
by using the continuous relaxation of Problem (\ref{eq:opt_prod}).
We follow the procedure described in Section \ref{sec:Algorithm}.
The MIQP error is depicted by a dashed line in Figure \ref{fig:MIQP1}.
To estimate the effect of numerical errors, we also calculated the
QP error which is shown in Figure \ref{fig:MIQP1} as a solid line. 

\begin{figure}
\begin{centering}
\includegraphics[bb=35bp 180bp 545bp 600bp,clip,scale=0.43]{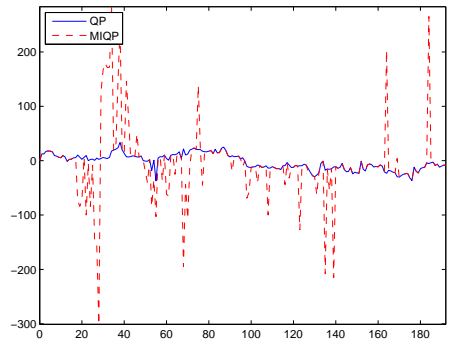}
\par\end{centering}

\caption{\label{fig:MIQP1}Error caused by considering Problem \ref{eq:opt_prod_rel}
instead of Problem \ref{eq:opt_prod}.}
\end{figure}
We can see from Figure \ref{fig:MIQP1} that $\left\Vert \sum_{c\in C}V_{c}^{*}+\sum_{p\in P}V_{p}^{*}\right\Vert _{\infty}\approx400\:\text{MWh}$.
Also in reality, production and consumption do not match exactly.
The mismatch is reflected through changes in the power line frequency.
In the UK, the nominal power line frequency is 50 Hz. The grid operator,
called National Grid, is responsible for keeping the frequency within
$\pm1\%$%
\footnote{See http://www2.nationalgrid.com/uk/services/balancing-services/frequency-response/.%
} of the nominal power line frequency. We can see from Figure \ref{fig:MIQP1}
that the largest errors occur at times when demand is high. Since
the overall demand for electricity during the peak hours is approximately
$40\:\text{GW}$ we can conclude that the error is within $\pm1\%$
error bound.

The model presented in this paper neglects the losses of electricity
in transmission and distribution lines. According to the World Bank%
\footnote{See http://data.worldbank.org/indicator/EG.ELC.LOSS.ZS/countries/GB?display=graph.%
} the transmission and distribution losses in the UK account for approximately
7.5\% (maximum 8.5\% in 2004 and minimum 7.0\% in 2010) of the total
electricity production. The losses vary in time and can change for
$\pm1\%$.

Due to the reasons above, we believe that for the purpose of modeling
realistic power prices, it is enough to consider the continuous relaxation
of Problem (\ref{eq:opt_prod}) and neglect binarity constraints.

\section{Grid operator\label{sec:Grid-operator}}

In Section \ref{sec:Problem-description}, we investigated how to
include startup costs in our model. The calculated equilibrium price
contained many spikes, which are in reality prevented by intervention
of the grid operator. In times, when the standing reserve is low,
the grid operator incentivizes additional power plants to turn on
and thus help making the delivery of electricity more reliable. In
this section we investigate how to incorporate the actions of the
grid operator into our model.

\subsection{Quadratic programming formulation}

The costs of the grid operator's actions that help to maintain a high
reliability of the delivery of electricity are distributed among all
market participants. All market participants are collectively penalized
in the situations when the standing reserve is low. To include the
penalization in our model, we propose a quadratic penalty function
$\Upsilon\left(SR\left(T_{j}\right)\right)$ defined as
\begin{equation}
\Upsilon\left(SR\left(T_{j}\right)\right):=\alpha\left(\max\left\{ 0,\beta-SR\left(T_{j}\right)\right\} \right)^{2},\label{eq:penalty}
\end{equation}
where $\alpha>0$ and $\beta>0$ are used to describe a risk aversion
of the grid operator. Parameter $\beta$ tells us at what level of
the standing reserve does the grid operator start to take action.
Parameter $\alpha$ tells us how much is the grid operator willing
to incentivize the power plant to start production.

One can incorporate the grid operator into Problem (\ref{eq:op_all-2})
as 
\begin{equation}
\begin{array}{cl}
\underset{x}{\text{max}} & -\pi^{\top}x-\frac{1}{2}x^{\top}Qx-\sum_{j\in J}\Upsilon\left(SR\left(T_{j}\right)\right)\\
\\
\text{s.t.} & Ax=a\\
\\
 & Bx\leq b\\
\\
 & \mu_{M}=0
\end{array}\label{eq:op_all-2-1}
\end{equation}

It might not be immediately clear, how to write the penalty term (\ref{eq:penalty})
in a quadratic programming framework. We can follow an approach that
is widely used in the linear programming literature and introduce
a decision variable $z\left(T_{j}\right)$ with the following constraints
\begin{equation}
\begin{array}{rcl}
z\left(T_{j}\right) & \geq & 0\\
\\
z\left(T_{j}\right) & \geq & \beta-SR\left(T_{j}\right),
\end{array}
\end{equation}
which hold for each $j\in J$. The penalty function $\Upsilon\left(SR\left(T_{j}\right)\right)$
can be written as a function of $z\left(T_{j}\right)$ as 
\begin{equation}
\Upsilon\left(z\left(T_{j}\right)\right):=\alpha z\left(T_{j}\right)^{2},\label{eq:penalty-1}
\end{equation}
which fits into the quadratic programming framework.

One can apply the procedure described in Section \ref{sec:Algorithm}
to find a more convenient dual formulation of Problem (\ref{eq:op_all-2-1}).

\subsection{Numerical results\label{sub:N2}}

In this section we investigate numerical results after inclusion of
the grid operator. Calculated equilibrium prices and the energy mix
and depicted in Figure \ref{fig:f8-1-2}. The figure on the left hand
side depicts the calculated energy mix between coal and gas power
plants, while the figure on the right hand side depicts the actually
observed energy mix. Both figures contain also the spot price calculated
by our model and the actually observed spot price. We set $\alpha=0.01$
and $\beta=1500$. Determination of the optimal standing reserve is
a challenging problem, which has received a lot of attention in the
literature (see \cite{ela2010evolution} and \cite{devos2014dynamic}
for example) and exceeds the scope this paper.

\begin{figure}
\begin{centering}
\includegraphics[bb=35bp 180bp 545bp 600bp,clip,scale=0.43]{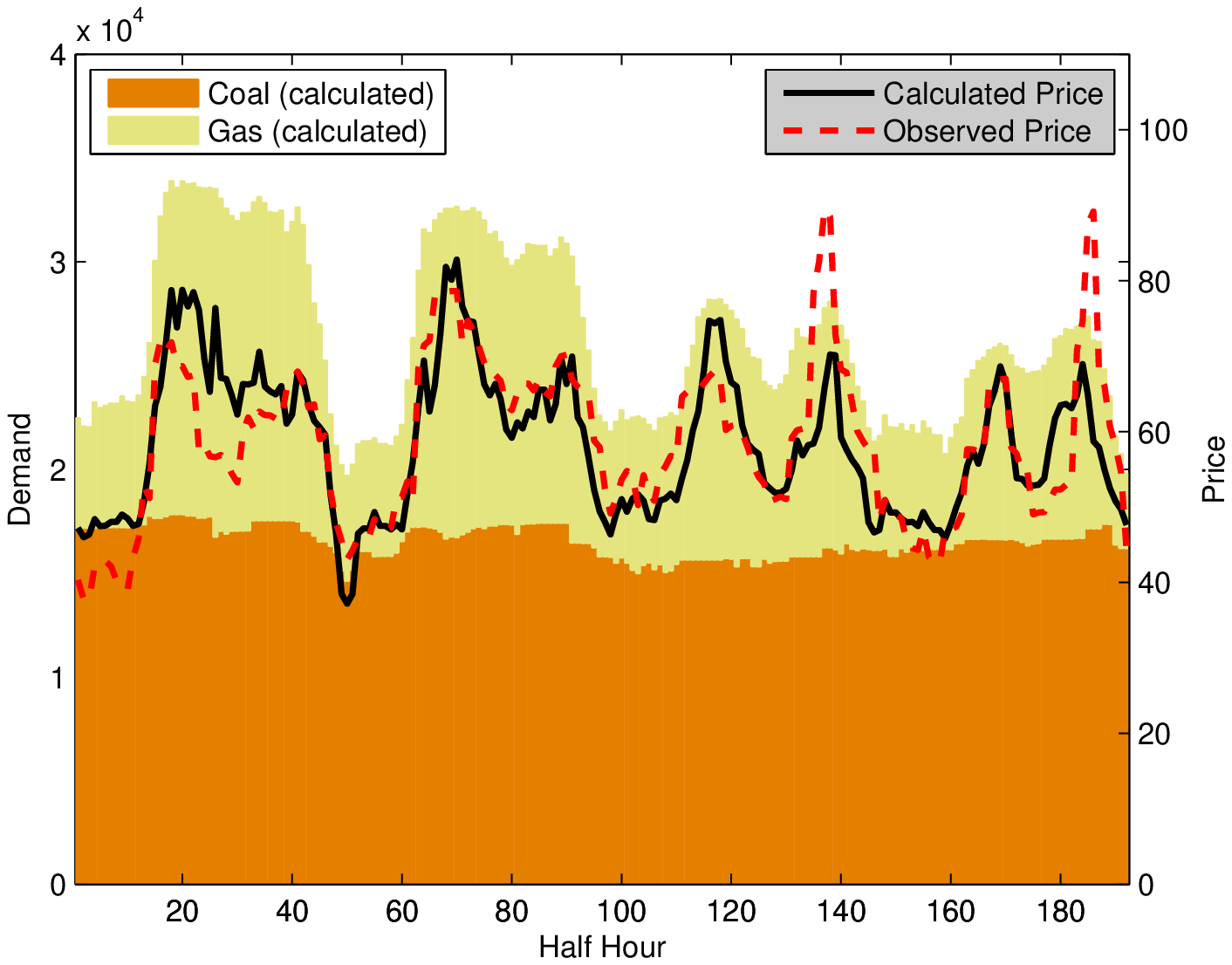}\includegraphics[bb=35bp 180bp 545bp 600bp,clip,scale=0.43]{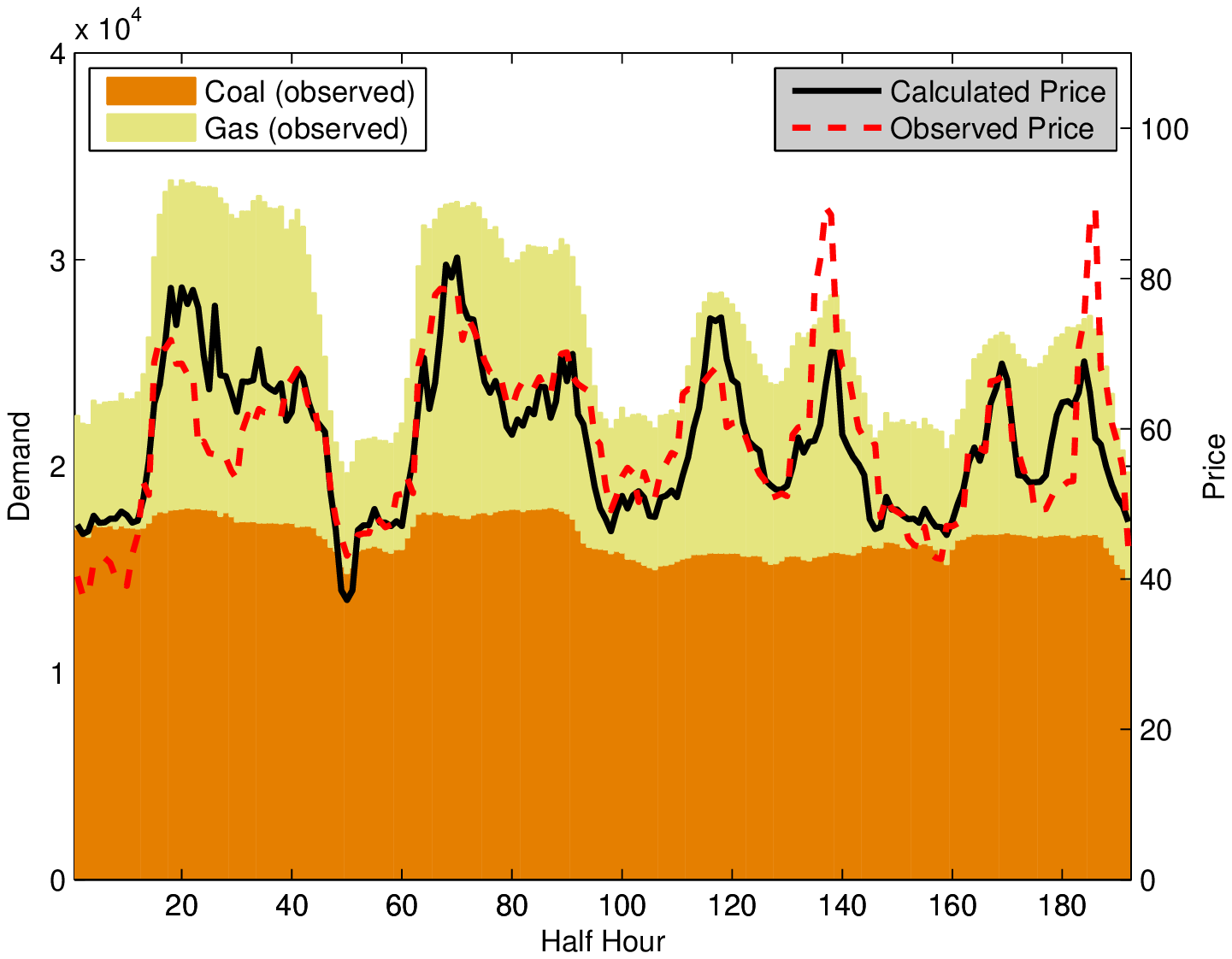}
\par\end{centering}

\caption{\label{fig:f8-1-2}Comparison of the calculated and historical electricity
price and energy mix with startup costs and the grid operator included.}
\end{figure}
\begin{figure}
\begin{centering}
\includegraphics[scale=0.43]{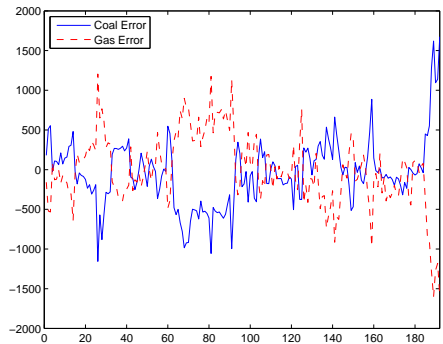}
\par\end{centering}

\caption{\label{fig:f5_1-1-1}The difference between calculated and observed
gas and coal production after including the grid operator.}
\end{figure}
By comparing Figure \ref{fig:f8-1} and Figure \ref{fig:f8-1-2},
we can see that the calculated equilibrium electricity price in Figure
\ref{fig:f8-1-2} follows the daily variations much more closely.
The calculated equilibrium electricity price does not contain any
spikes, because the grid operator prevented them by managing the standing
reserve. In our model, we assume that the players (and the grid operator)
have a perfect demand forecast. However, in reality this is usually
not the case. The grid operator is not able to predict the demand
perfectly, and corrective actions are often required. When large corrective
action is required at times close to delivery, then only a few (usually
rather inefficient Open Cycle Gas Turbine) power plants are flexible
enough to cover the demand, which causes spikes in the electricity
price. Modeling of recursive actions exceeds the scope of this paper
and is left for future work. 

Figure \ref{fig:f5_1-1-1} shows that the inclusion of the grid operator
did not have any significant impact on the error in the energy mix.

Figure \ref{fig:f8-1-1-1} shows the standing reserve after inclusion
of the grid operator. The standing reserve never reaches zero since
the grid operator prevents this by requiring new power plants to start
production to ensure stability of the electricity grid. This makes
the spot price smoother and significantly decreases the number of
spikes. 

\begin{figure}
\begin{centering}
\includegraphics[bb=35bp 180bp 545bp 600bp,clip,scale=0.43]{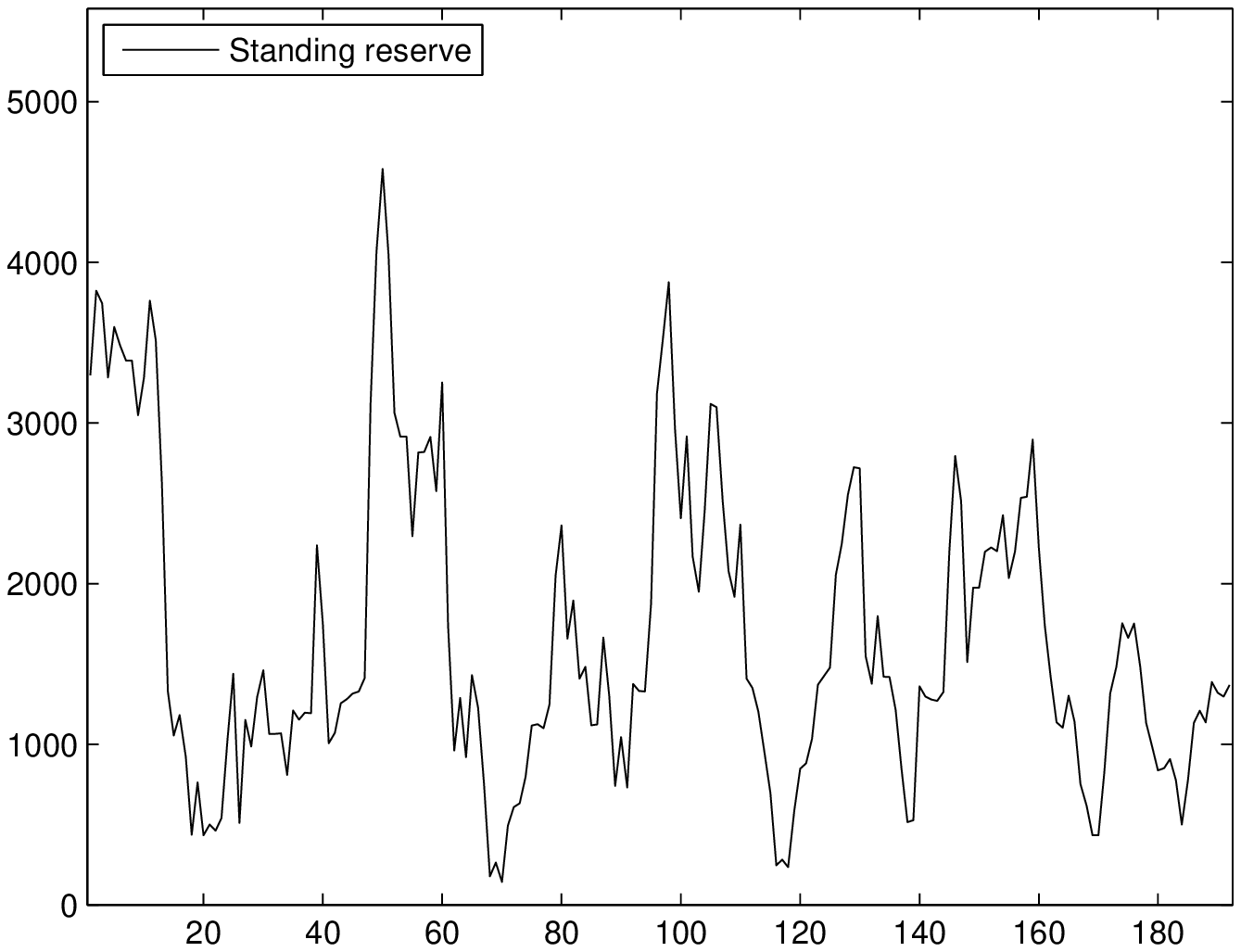}
\par\end{centering}

\caption{\label{fig:f8-1-1-1}Standing reserve for the relevant time period.}
\end{figure}

Figure \ref{fig:MIQP2} depicts the MIQP and QP errors after inclusion
of the grid operator. By comparing Figure \ref{fig:MIQP1} and Figure
\ref{fig:MIQP2}, we can see that the inclusion of the grid operator
has a small impact on the errors, which remained within $\pm1\%$
error bound.

\begin{figure}
\begin{centering}
\includegraphics[bb=35bp 180bp 545bp 600bp,clip,scale=0.43]{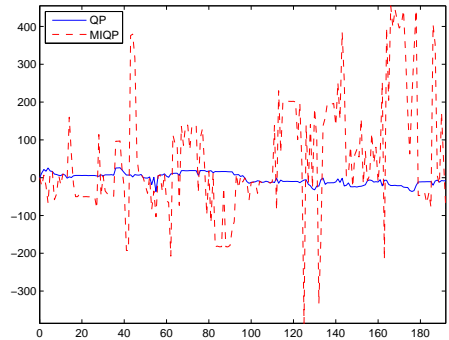}
\par\end{centering}

\caption{\label{fig:MIQP2}Error caused by considering Problem \ref{eq:opt_prod_rel}
instead of Problem \ref{eq:opt_prod}.}

\end{figure}

\section{Conclusions\label{sec:Conclusions}}

In this paper we proposed a tractable quadratic programming formulation
for calculating the equilibrium term structure of electricity prices
when the startup costs of power plants are included in the model.
Through numerical simulations we showed that startup costs have a
large impact on electricity prices. When startup costs are included
in the model, the calculated spot electricity price during peak hours
increased and during off-peak hours decreased. Moreover, startup costs
are responsible for introducing frequent high spikes in the spot electricity
price.

We observed that price spikes occur at times when the standing reserve
in low. In reality, the times of a low standing reserve are rare,
because of the intervention of the grid operator, who is responsible
for providing a reliable electricity delivery and preventing times
with a low standing reserve. We included the grid operator in our
model in the second part of the paper. This significantly decreased
the number of spikes. Moreover, the computed equilibrium electricity
prices matched the historically observed prices very closely. 

Numerical simulations were performed by modeling the realistic UK
power grid consisting of a few hundred power plants. A tractable approach
to estimate startup costs of power plants from their historical production
was also proposed. 

\bibliographystyle{siam}
\bibliography{mat_fin_bib,transfer}

\end{document}